%% file: main.tex
\documentclass[letterpaper, 10 pt, conference]{ieeeconf}  

\IEEEoverridecommandlockouts                              
                                                          
\usepackage{epsfig}
\usepackage{url}
\usepackage{amstext}
\usepackage{amsmath}
\usepackage{amsfonts}

\usepackage{amsthm}
\usepackage{amssymb}
\usepackage{graphicx}
\usepackage{color}
\usepackage{balance} 
\usepackage{algorithm2e}
\usepackage{algorithmic}
\usepackage{bm}
\usepackage{tabularx}
\usepackage{amssymb,amsmath,amsthm} 
\usepackage{multicol}
\usepackage{multirow}

\newtheorem{definition}{Definition}

\let\labelindent\relax

\usepackage{enumitem}
\usepackage{graphicx} 
\graphicspath{{figure/}}
\usepackage{float}
\usepackage[caption=false,font=footnotesize]{subfig}
\usepackage{soul}
\usepackage{color}

\usepackage{algorithm,algorithmic}
\usepackage{cite}

\newtheoremstyle{bfnote}%
{}{}%
{\itshape}{}%
{\bfseries}{.}%
{ }%
{\thmname{#1}\thmnumber{ #2}\thmnote{ (#3)}}
\theoremstyle{bfnote}
\newtheorem{thm}{Theorem}

\newtheorem{lemma}{Lemma}

\usepackage{tikz}

\usepackage{nicematrix}
\usepackage{algorithm,algorithmic}
\usepackage{cite} 
\usepackage{url}            

\usepackage{xcolor}         
\usepackage[breaklinks=true,letterpaper=true,colorlinks,bookmarks=false]{hyperref}
\hypersetup{citecolor=[RGB]{0,0,255}}

\title{\LARGE \bf Equilibrium-Independent  Stability Analysis for Distribution Systems with Lossy Transmission Lines}

\author{Wenqi Cui and Baosen Zhang
\thanks{Department of Electrical and Computer Engineering, University of Washington Seattle, WA 98195, USA	 \{wenqicui, zhangbao\}@uw.edu}%
\thanks{The authors are supported in part by the National Science Foundation grants ECCS-1930605 and ECCS-2153937.}}

\begin{document}
\maketitle
\thispagestyle{empty}
\pagestyle{empty}

\begin{abstract}
Power distribution systems are becoming much more active with increased penetration of distributed energy resources. Because of the intermittent nature of these resources, the stability of distribution systems under large disturbances and time-varying conditions is becoming a key issue in practical operations. Because the transmission lines in distribution systems are lossy, standard approaches in power system stability analysis do not readily apply and the understanding of transient stability remains open even for simplified models.

This paper proposes a novel equilibrium-independent transient stability analysis of distribution systems with lossy lines. We certify network-level stability by breaking the network into subsystems, and by looking at the equilibrium-independent passivity of each subsystem,  the network stability is certified through a diagonal stability property of the interconnection matrix. This allows the analysis scale to large networked systems with time-varying equilibria. The proposed method gracefully extrapolates between lossless and lossy systems, and provides a simple yet effective approach to optimize control efforts with guaranteed stability regions. Case studies verify that the proposed method is much less conservative than existing approaches and also scales to large systems.


\end{abstract}


\section{Introduction}
\input{introduction}

\section{Model and Problem Formulation} \label{sec:model}
\input{model}

\section{Modular Design of Subsystems} \label{sec:module}
\input{module}

\section{Compositional Stability Certification} \label{section:Stability}
\input{Stability}

\section{ Controller Design from EIP of Subsystems} \label{section:Modular_EIP}
\input{EIP}

\section{Case Study} \label{sec:simulation}
\input{simulation}

\section{Conclusion} \label{sec:conclusion}
This paper proposes a modular approach for transient stability analysis of distribution systems with lossy transmission lines and time-varying equilibria. Network stability is decomposed into the strictly EIP of subsystems and the diagonal stability of the interconnection matrix. This in turn provides a simple yet effective approach to
optimize damping coefficients with guaranteed stability regions. Case studies  show that the proposed method is less conservative compared with existing approaches and can scale to large systems. The Pareto-front for the trade-off between
stability regions and control efforts can also be efficiently computed.

\bibliographystyle{IEEEtran}
\bibliography{Reference}

\section{Appendix} \label{sec:appendix}
\input{appendix}
\end{document}

%% file: introduction.tex
Distributed energy resources (DERs) such as rooftop solar, electric vehicles and battery storage devices are increasingly entering the power distribution systems. These devices have intermittent outputs and often exhibit large and fast ramping variations, bringing larger disturbances to the system~\cite{xu2019data,ross2020method}.
Therefore, stability of distribution systems under time-varying conditions and large disturbances is becoming a key question in their operations~\cite{zhang2016transient}.



We are mainly interested in the ability of a system to converge to an acceptable equilibrium following large disturbances~\cite{chiang1989study,sastry2013nonlinear}. In power systems, this is often called transient stability analysis. Most of the time,  transmission lines\footnote{ Power lines in the distribution system is also called transmission lines.} are assumed to be lossless (i.e., the lines are purely inductive with zero resistances). This significantly simplifies the mathematical analysis and allows for explicit constructions of energy functions for  microgrids~\cite{schiffer2014conditions, de2017bregman}, transmission systems~\cite{cui2021reinforcement, sauer2017power} and network-preserved differential-algebraic models~\cite{chiang2011direct, de2016lyapunov}. However, the transmission lines in distribution systems have non-negligible resistances~\cite{kersting1991radial}. More precisely, the $r/x$ ratios of the lines are not very small and the lines are called ``lossy''~\cite{robbins2012two,Zhang15}. For lossy systems, transient stability becomes a much harder problem and remains open even for simplified models~\cite{ zhang2016transient, huang2021neural}. 


A main difficulty in transient stability analysis for lossy networks is the lack of a good Lyapunov function (or energy function)~\cite{chiang1989study,sastry2013nonlinear}. 
Existing explicit constructions require all the lines to have the same $r/x$ ratios~\cite{de2017bregman}. In more general cases, a classical approach is to use path-dependent integrals to construct Lyapunov functions, but these integrals are not always well-defined and rely on knowing the trajectories of the states~\cite{chiang1989study}.
 Some works use linear matrix inequalities (LMIs) to find Lyapunov functions by relaxing sinusoidal AC power flow equations~\cite{zhang2016transient,vu2015lyapunov}. These relaxations bound sinusoidal functions with linear or quadratic ones, but the bound can be loose and lead to conservative stability assessments. A candidate Lyapunov function can also be found via Sum Of Squares (SOS) programming techniques~\cite{anghel2013algorithmic}, but the computation complexity grows quickly with increased problem size. This makes the method difficult to scale to moderate or large systems. More recently, attempts have been made to learn a Lyapunov function parameterized by neural networks~\cite{huang2021neural, cui2021lyapunov}. However, it is challenging to verify that the learned neural networks are actually Lyapunov functions.


Apart from the challenges in scalability, existing approaches only apply a single equilibrium at a time~\cite{zhang2016transient,huang2021neural, miao2019modeling}. Because of frequent  changes to DERs' setpoints,  equilibria are time-varying. Hence, it is essential to characterize stability for a set of possible equilibria. In addition,  the power electronics on the DERs allow their damping coefficients to be adjusted~\cite{johnson2013synchronization,cui2021lyapunov}. But optimizing these coefficients using existing approaches are nontrivial, since they involve solving complicated nonconvex problems. Therefore, the coefficients often are tuned slowly by trial and error, making the design process cumbersome and difficult. 

This paper proposes a novel equilibrium-independent approach to transient stability analysis of lossy distribution systems, where we achieve scalability by breaking the network into subsystems. In particular, we consider the angle droop control for the power-electronic interfaces to drive voltage phase angles to their setpoints~\cite{zhang2016transient, huang2021neural}. For lossy transmission lines, we design a tunable parameter that can serve to explicitly trade off between the control effort and the stability region. At the limit, we recover results for lossless transmission lines, allowing the proposed method to gracefully extrapolates between lossless and lossy systems.

 Motivated by  equilibrium-independent passivity (EIP)  proposed in~\cite{hines2011equilibrium, arcak2016networks}, we study the network stability with time-varying equilibrium points by certifying EIP of each subsystems. 
 Then, stability certification is reduced to checking the diagonal stability property of the interconnection matrix over subsystems subject to EIP conditions. The proposed design of the subsystems divides the interconnection matrix into the summation of a skew-symmetric and a sparse matrix. The stabilizing damping coefficients are then explicitly represented as a convex constraint.
 This in turn provides a simple yet effective approach to optimize control efforts with guaranteed stability regions. 
Case studies verify that the proposed method is much less conservative and much more scalable to large systems compared with existing methods~\cite{zhang2016transient, huang2021neural}.

%% file: model.tex
\subsection{Power-Electronic Interfaced Distribution Systems}


  Consider a distribution system with $n$ buses and $m$ lines modelled as a connected graph $\left(\mathcal{N},\mathcal{L} \right)$, where each bus is equipped with
  with a power-electronic interface~\cite{zhang2016transient, huang2021neural}
Buses are indexed by $k \in \mathcal{N} :=\{1,\dots, n\} $. Lines are indexed by  $l\in\mathcal{L}:=\left\{n+1,\cdots,n+m\right\}$. 
Without loss of generality, we define the power flow from $i$ to $j$ to be the positive direction if $i<j$. 
We denote the interconnections between buses $i, j$ and line $l$ connecting them as $l \in \mathcal{B}_{i}^{+}$ and  $l\in \mathcal{B}_{j}^{-}$, where $\mathcal{B}_{i}^{+}$
and $\mathcal{B}_{j}^{-}$
  represents the line $l$ leaving bus $i$ and entering bus $j$, respectively.

We adopt the model proposed in~\cite{zhang2016transient} where angle and voltage droop control are utilized for
real and reactive power sharing through power-electronic interfaces. Let $\delta_{k}$ and $v_{k}$ be the voltage phase angle and voltage magnitude at bus  $k\in\mathcal{N}$, and $\delta_{k}^{*},v_{k}^{*}$ be their setpoint values set by distribution system operators~(for more information on how the setpoints are chosen, see~\cite{zhang2016transient, huang2021neural}).  Let $p_{k}$ and $q_{k}$ denote real and reactive power injections at bus $k$, and $p_{k}^{*}$ and $q_{k}^{*}$ be their setpoints.  The dynamics of bus $k$ are described by
\begin{subequations}\label{eq:dynamics}
\begin{align}
\tau_{\mathrm{a} k} \dot{\delta}_{k}&=
-d_{ak}(\delta_{k}-\delta_{k}^\ast)+ \left(p_{k}^{*}-p_{k}\right)\label{subeq:dynamics_delta} \\
\tau_{\mathrm{v} k} \dot{v}_{k}&=-d_{vk}(v_{k}-v_{k}^\ast)+ \left(q_{k}^{*}-q_{k}\right)\label{subeq:dynamics)v},
\end{align}
\end{subequations}
where $\tau_{\mathrm{a} k}$ and $\tau_{\mathrm{v} k}$ are time constants for voltage phase angle and voltage magnitude at bus $k$, respectively. The parameters $d_{ak}$ and $d_{vk}$ are damping coefficients controlling power injected by inverters, and thus larger values correspond to larger control efforts. Importantly, the equilibria of the system come from the setpoints $\delta^*$ and $v_k^*$, which are time varying and not known ahead of time.

We follow the model in~\cite{zhang2016transient, huang2021neural} where $\tau_{vk}\gg \tau_{ak}$ by design. Then, the voltage $v_k$ evolves much slower than the phase angle $\delta_k$, hence the angle and voltage dynamics separates in timescale and $v_k$ is typically assumes to be constant. We therefore focus on the angle stability dynamics in~\eqref{subeq:dynamics_delta}  and set $v_k=1$ per unit in the rest of this paper.

Let $g_l$ and $b_l$ be the conductance and susceptance of the transmission line $l\in\mathcal{L}$, respectively. The active power flow in the line $l$ from bus $i$ to $j$  is
\begin{equation}\label{eq:dynamics_line}
p_l=g_l-g_l\cos(\delta_i-\delta_j)+b_l\sin(\delta_i-\delta_j),  
\end{equation}
which is the nonlinear AC power flow equations. We often use $\delta_{ij}$ as a shorthand for $\delta_i-\delta_j$. System operators calculate the setpoints such that $p_k^*$ and $\delta_k^*$ satisfy the power flow equation for all $k\in\mathcal{N}$.
A transmission line is called lossless if $g_l=0$ and lossy otherwise. For distribution systems, $g_l$ is typically not significantly smaller than $b_l$. 



The buses are interconnected with transmission lines and the active power injected from bus $k$ to the network is
\begin{equation}\label{eq:power_injection}
p_k = \sum_{l\in  \mathcal{B}_k^+} p_l-\sum_{l\in  \mathcal{B}_k^-}p_l \quad .
\end{equation}
The dynamics of the system is described by~\eqref{subeq:dynamics_delta},~\eqref{eq:dynamics_line} and~\eqref{eq:power_injection}. 
The transient stability of the system is defined as the ability to converge to the equilibrium points $\delta^*$ from different initial conditions. Since equilibria are set by system operators, the system needs to be stable for multiple possible equilibria. In this paper, we adopt a modular approach to certify stability and design the damping coefficients $d_{ak}$'s, and show how it overcomes the challenges of existing approaches. 

\subsection{Stability Analysis Through A Modular Approach }

The goal of this paper is to answer two key questions for the transient stability of distribution systems: \textit{1) How large is the stability region?} and \textit{2) What is the control effort needed to attain certain range of stability region?} To this end, we certify network-level stability by breaking the network into subsystems. Then by looking at the equilibrium-independent passivity (EIP) of each subsystems and their interconnections, the stability analysis scale to large networked systems with time-varying equilibrium points~\cite{arcak2016networks}. 


For each bus~\eqref{subeq:dynamics_delta} and each transmission line~\eqref{eq:dynamics_line}, we abstract them as a subsystem $G_i$ with input $\bm{u}_i$ and output $\bm{y}_i$. Fig.~\ref{fig:Interconnection} shows the diagram for the connection of subsystems. 
The coupling of the input and output of each subsystems are described by $\bm{u}=\bm{M}\bm{y}$, where the matrix $\bm{M}$ is determined by interconnections of the system. 
We show that $\bm{M}$ is the summation of a skew-symmetric matrix $\bm{M}_1$ and a sparse matrix $\bm{M}_2$. This enable us to obtain a compact and convex expression of stabilizing damping coefficients, which can easily be used for controller design. 

Our method gracefully extrapolates between lossless and lossy systems. If all the lines lossless, the sparse component of $\bm{M}_2$  is zero and only the skew-symmetric part remains. Then standard results from EIP theory can be used to directly show the stability of the system, illustrating why lossless systems are simpler than lossy ones.

\begin{figure}[ht]	
	\centering
	\includegraphics[width=2.6in
	]{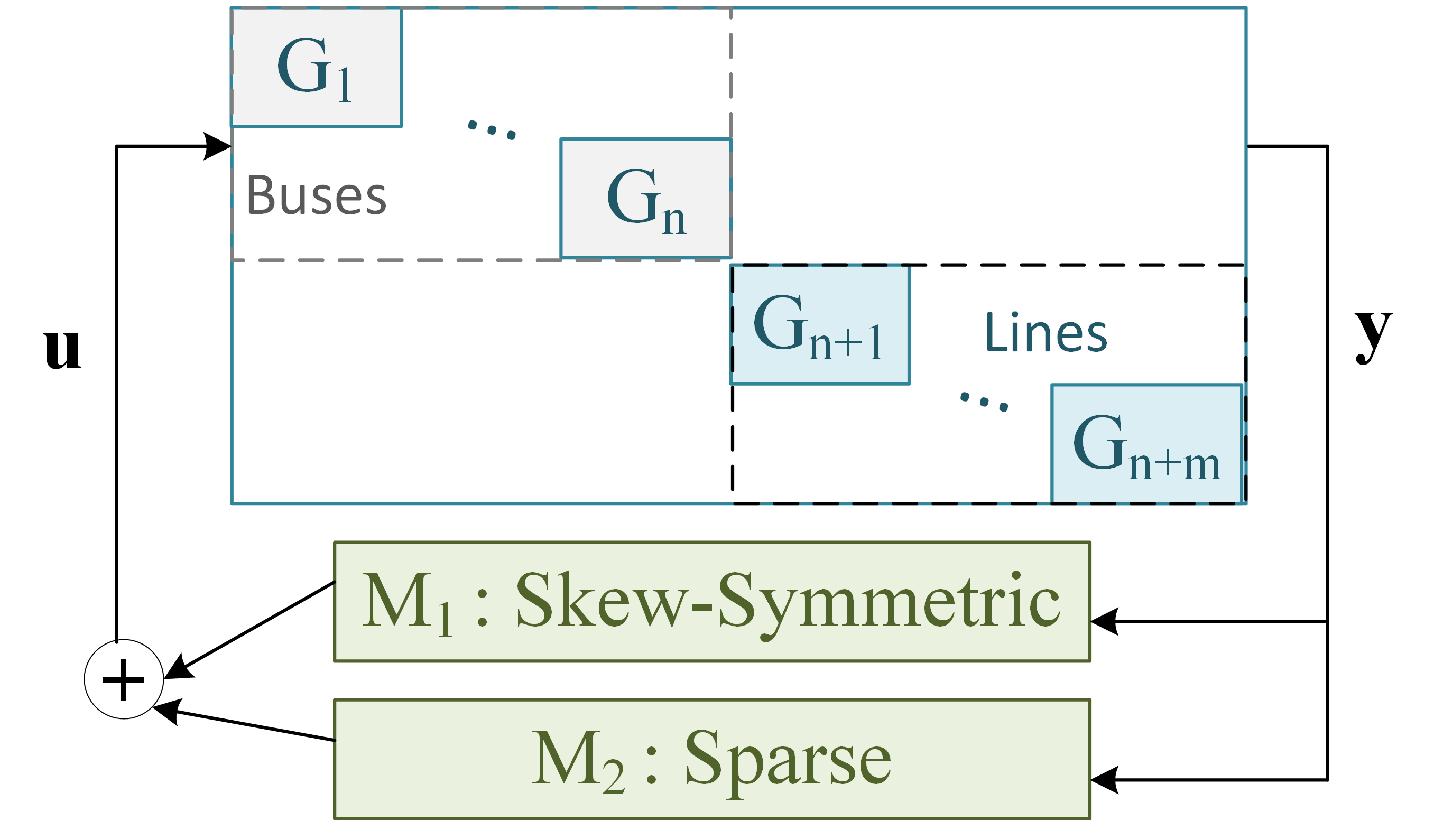}
	\caption{Interconnection of buses (grey blocks) and transmission lines (blue blocks). The input and output of each subsystems are interconnected through the $\bm{y}=(\bm{M}_1+\bm{M}_2)\bm{u}$, where $\bm{M}_1$ is skew-symmetric and $\bm{M}_2$ is sparse.}
	\label{fig:Interconnection}
\end{figure}

%% file: module.tex
With the aim of network stability assessment through the passivity of subsystems, we study the abstraction of~\eqref{subeq:dynamics_delta}-\eqref{eq:power_injection} as subsystems of buses and lossy transmission lines and their input-output interconnections in this section. 

\subsection{Subsystems for Buses and Lossy Transmission Lines}


The subsystem for the lossy transmission line $l\in \mathcal{L}$ leaving bus $i$ and entering bus $j$ is defined with the input $\bm{u}_{l}=\begin{bmatrix}  \delta_i-\delta_j &  \delta_j-\delta_i \end{bmatrix}^\top
\in \mathbb{R}^2$ to be the angle differences  from $i$ to $j$ and from $j$ to $i$. The output $\bm{y}_{l}\in \mathbb{R}^2$ is defined to be the modified power flow from $i$ to $j$ and from $j$ to $i$:
\begin{subequations}\label{eq:Modular_Line}
\begin{align}
\left [\begin{matrix}
     y_{l,1}\\y_{l,2}
    \end{matrix}\right ] &=\frac{1}{2}
    \left [\begin{matrix}
    \left ( g_l-g_l\cos(u_{l,1}) \right )/\alpha_l+b_l \sin(u_{l,1})\\
     \left ( g_l-g_l\cos(u_{l,2}) \right )/\alpha_l+b_l \sin(u_{l,2})
    \end{matrix}\right ] \\
    \begin{split}
    \left [\begin{matrix}
     u_{l,1}\\u_{l,2}
    \end{matrix}\right ]
    &=
    \underbrace{\left[\begin{matrix}
     1 \\-1
    \end{matrix}\right ]}_{\Phi_{li}} \delta_i
    +
    \underbrace{\left[\begin{matrix}
     -1 \\1
    \end{matrix}\right ]}_{\Phi_{lj}} \delta_j
    \end{split}
\end{align}
\end{subequations}
where $\alpha_l>0$ is a tunable scalar and we will study later in detail. At a high level, a  larger $\alpha_l$  implies larger stability regions and larger stabilizing damping coefficients. The power flow~\eqref{eq:dynamics_line} from bus $i$ to $j$ and that from bus $j$ to $i$  can be recovered by $p_{ij}=y_{l,1}-y_{l,2}+\alpha_l(y_{l,1}+y_{l,2})$ and $p_{ji}=-y_{l,1}+y_{l,2}+\alpha_l(y_{l,1}+y_{l,2})$,
which will then serve as the input to the subsystem of  buses. 
Stacking the inputs and outputs of lines gives $\bm{u}_{\mathcal{L}}=[\bm{u}_{n+1}^\top \cdots \bm{u}_{n+m}^\top]\in\mathbb{R}^{2m}$, and $\bm{y}_{\mathcal{L}}=[\bm{y}_{n+1}^\top \cdots \bm{y}_{n+m}^\top]\in\mathbb{R}^{2m}$. The matrix block $\Phi_{li}:=\begin{bmatrix}  1&  -1 \end{bmatrix}^\top$ and  $\Phi_{lj}:=\begin{bmatrix}   -1& 1 \end{bmatrix}^\top$ are defined for the mapping from the output of the head $i$ and the tail $j$ to the input of line $l$, respectively.

The subsystem for bus $k$ is defined with the input $u_k\in\mathbb{R}$ to be the power injection from connected transmission lines and the output $y_k\in\mathbb{R}$ to be the phase angle
\begin{subequations}\label{eq:Modular_Generator}
\begin{align}
\tau_{ k} \dot{\delta}_{k}& =-d_k(\delta_{k}-\delta_{k}^\ast)+ \left(P_{k}^{*}+u_{k}\right) \\
y_{ k} &=\delta_{k}\label{subeq:Modular_Generator_y}\\
  \begin{split}
      u_{ k} &=  \sum_{l\in  \mathcal{B}_k^+}
      \underbrace{\left[\begin{matrix}
     -1 & 1
    \end{matrix}\right ]}_{\Phi_{kl}} y_l
    +\underbrace{\alpha_l\left[\begin{matrix}
     -1 & -1
    \end{matrix}\right ]}_{\Psi_{kl}} y_l\\
    &\quad
    +\sum_{l\in  \mathcal{B}_k^-}
    \underbrace{\left[\begin{matrix}
     1 & -1
    \end{matrix}\right ]}_{\Phi_{kl}} y_l
    +\underbrace{\alpha_l\left[\begin{matrix}
     -1 & -1
    \end{matrix}\right ]}_{\Psi_{kl}} y_l
    \end{split}
\end{align}
\end{subequations}
where the matrix block  $\Phi_{kl}$ and $\Psi_{kl}$ is defined for the mapping from the output of the subsystem for line $l\in\mathcal{L}$ to the input of the subsystem for bus $k\in\mathcal{N}$. The matrix block $\Phi_{kl}:=\begin{bmatrix}   -1& 1 \end{bmatrix}$ if $l \in \mathcal{B}_k^{+}$ and  $\Phi_{kl}:=\begin{bmatrix}  1&  -1 \end{bmatrix}$ if  $l\in \mathcal{B}_k^{-}$.  The matrix block $\Phi_{kl}:=\begin{bmatrix} -\alpha_k& -\alpha_k \end{bmatrix}$ is defined uniformly for all line $l$ that connects bus $k$. It will serve to constrain the minimum-effort damping coefficients that stabilize the system. More details of the modular design and how it recovers the original dynamics can be find in Appendix~\ref{app: modular}.

\subsection{The Interconnection of Subsystems}

To investigate the stability of the whole interconnected system, we stack the input/output vectors in sequence as $\bm{u}:=(\bm{u}_{\mathcal{N}}, \bm{u}_{\mathcal{L}})\in\mathbb{R}^{n+2m}$ and $\bm{y}:=(\bm{y}_{\mathcal{N}}, \bm{y}_{\mathcal{L}})\in\mathbb{R}^{n+2m}$.  
The mapping from the output of the bus $k\in\mathcal{N}$ to the input of the line $l\in\mathcal{L}$ is described by a matrix $\bm{\Phi}_{\mathcal{LN}}\in \mathbb{R}^{ 2m\times n}$, where the block in the $(2l-1)$-th, $2l$-th row and the $k$-th column  is $\Phi_{lk}$ in~\eqref{eq:Modular_Line}. Similarly,
the mapping from the output of the line $l\in\mathcal{L}$ to the input of the bus $k\in\mathcal{N}$ is described by the matrix $\bm{\Phi}_{\mathcal{NL}}\in \mathbb{R}^{n \times 2m}$, where the block in the $k$-th row and the $(2l-1)$ to $2l$-th column is $\Phi_{kl}$ in~\eqref{eq:Modular_Generator}. The input-output dependent on $\alpha$ is represented in the matrix $\bm{\Psi}\in \mathbb{R}^{n \times 2m}$, where the block in the $k$-th row and the $(2l-1)$ to $2l$-th column  is $\Psi_{kl}$ in~\eqref{eq:Modular_Generator}.
Then, the interconnection of subsystems represented in~\eqref{eq:Modular_Line} and~\eqref{eq:Modular_Generator} are compactly described by
\begin{equation}\label{eq:Matrix_M}
\bm{u}=\left (\bm{M}_1+\bm{M}_2\right )\bm{y}    
\end{equation}
where
$$
\bm{M}_1:=
    \begin{bmatrix}
\bm{0}_{n\times n}  & \bm{\Phi}_{\mathcal{NL}} \\  \bm{\Phi}_{\mathcal{LN}}   & \bm{0}_{2m\times 2m}
\end{bmatrix}, 
\bm{M}_2:=
\begin{bmatrix}
\bm{0}_{n\times n}  & \bm{\Psi } \\ \bm{0}_{2m\times n}   & \bm{0}_{2m\times 2m}
\end{bmatrix}.
$$

Note that the matrix $\bm{\Phi}_{\mathcal{NL}}$ and $\bm{\Phi}_{\mathcal{LN}}$ is constituted by the blocks that satisfy $\Phi_{il} = -\Phi_{li}^\top$ for all $i\in\mathcal{N}$ and $l\in\mathcal{L}$, we have $\bm{\Phi}_{\mathcal{NL}} +\bm{\Phi}_{\mathcal{LN}}^\top = \bm{0}$ and thus $\bm{M}_1$ is skew-symmetric. Two examples can be find in Appendix~\ref{app:example} to provide more details on how the proposed method works. The next section will show how the skew-symmetricity of $\bm{M}_1$ and the sparsity of $\bm{M}_2$ can be utilized for stability assessment of networked systems.

%% file: Stability.tex
\subsection{Stability Region}
The stability region is the set of initial states that converges to an  equilibrium. Formally, it is defined as~\cite{chiang2011direct}:
\begin{definition}[Stability Region]\label{def: stability}
A dynamical system $\dot{\bm{\delta}}=\bm{f}_{\bm{u}}(\bm{\delta})$ is asymptotically stable around an equilibrium $\bm{\delta}^*$ if, $\forall \varphi >0$, $\exists \theta>0$ such that $\|\bm{\delta}(0) - \bm{\delta}^*\|<\theta$ implies $\|\bm{\delta}(t) - \bm{\delta}^*\|<\varphi $ and $\lim_{t\rightarrow\infty}\bm{\delta}(t)\longrightarrow\bm{\delta}^*$.  
The stability region of a stable equilibrium $\bm{\delta}^*$ is the set of all states $\bm{\delta}$ such that
$
\lim_{t\rightarrow\infty}\bm{\delta}(t)\longrightarrow\bm{\delta}^*
$.
\end{definition}

For nonlinear systems, it is very difficult to characterize the exact geometry of the whole stability region. This paper, and most others (see, e.g.~\cite{chiang2011direct, schiffer2014conditions, de2017bregman}), attempt to find an inner approximation to the true stability region through Lyapunov's direct method. Correspondingly, the stability region is algebraically caulculated by the states satisfying Lyapunov conditions
$\mathcal{S}|_{V(\cdot)}=\left\{\bm{\delta}|V(\bm{\delta})\geq0, \Dot{V}(\bm{\delta})\leq 0\right\}$ with $V(\bm{\delta})$ be a  Lyapunov function that equals zero at equilibrium. 
In the next subsections, we construct a Lyapunov function from equilibrium-independent
passivity of subsystems, which will bring larger stability region than existing methods~\cite{zhang2016transient}.

\subsection{Equilibrium Independent Passivity }

Equilibrium-independent
passivity (EIP), characterized by a dissipation inequality
 referenced to an arbitrary equilibrium input/output pair, allows one
to ascertain passivity of the components without knowledge of
the exact equilibrium~\cite{hines2011equilibrium}. The definition is given as follows~\cite{hines2011equilibrium,arcak2016networks}:
\begin{definition}[Equilibrium-Independent
Passivity]
 The system described by
$\dot{\bm{\delta}}=f(\bm{\delta}, \bm{u}),
 \bm{y}=h(\bm{\delta}, \bm{u})$ 
 is equilibrium-independent 
passive in a set $\bm{\delta}\in \mathcal{S}$ if, for every possible equilibrium $\bm{\delta}^*\in \mathcal{S}$, there exists a continuously-differentiable storage function $V_{\bm{\delta}^*}: \mathcal{S} \rightarrow \mathbb{R}_{\geq 0}$, such that $V_{\bm{\delta}^*}(\bm{\delta}^*)=0$ and
\begin{equation*}
 \nabla_{\bm{\delta}} V_{\bm{\delta}^*}(\bm{\delta})^{T} f(\bm{\delta}, \bm{u}) \leq (\bm{u}-\bm{u}^*)^\top(\bm{y}-\bm{y}^*).   
\end{equation*}
    
If there further exists a positive scalar $\epsilon$ such that
\begin{equation}\label{eq:def_sEIP}
\begin{split}
\nabla_{\bm{\delta}} V_{\bm{\delta}^*}(\bm{\delta})^{T} f(\bm{\delta}, \bm{u}) \leq &
(\bm{u}-\bm{u}^*)^\top(\bm{y}-\bm{y}^*) \\
&-\epsilon (\bm{y}-\bm{y}^*)^\top(\bm{y}-\bm{y}^*),
\end{split}
\end{equation}
then the system is strictly EIP.
\end{definition}

 In Section~\ref{section:Modular_EIP}, we will show that subsystems~\eqref{eq:Modular_Generator} corresponding to the bus $k\in\mathcal{N}$ is strictly EIP {in the region $\mathcal{S}_k$} with $\epsilon_k=d_{ak}$
 and the storage function $V_k(\bm{\delta})=\frac{1}{2\tau_k}\left(\delta_k-\delta_k^*\right)^2$ . The subsystem~\eqref{eq:Modular_Line} corresponding to the line $l\in\mathcal{L}$ is strictly EIP {in the region $\mathcal{S}_l$} with $\epsilon_l= \frac{2\alpha_l}{\sqrt{g_l^2+b_l^2\alpha_l^2}}$ and the storage function $V_l(\bm{\delta}) = 0$ . We denote 
 $\bm{\epsilon}_\mathcal{N}:=\left(\epsilon_1,\cdots, \epsilon_n\right)$ 
 , $\bm{\epsilon}_\mathcal{L}:=\left(\epsilon_{n+1}\bm{1}_2,\cdots, \epsilon_{n+m}\bm{1}_2\right)$ for the EIP coefficients of buses and lines, and the diagonal matrices $\hat{\bm{\epsilon}}_\mathcal{L}:=\operatorname{diag}\left(\bm{\epsilon}_\mathcal{L}\right)$, $\hat{\bm{\epsilon}}_\mathcal{N}:=\operatorname{diag}\left(\bm{\epsilon}_\mathcal{N}\right)$ and $\hat{\bm{\epsilon}} := \operatorname{diag}\left(\bm{\epsilon}_\mathcal{N}, \bm{\epsilon}_\mathcal{L}\right)=$ that will be used in network stability certification. In particular, let  $\bm{d}_\mathcal{N}:=\left(d_{a1},\cdots, d_{an}\right)$, we have
$
\hat{\bm{\epsilon}}_\mathcal{N}=\operatorname{diag}\left(\bm{d}_\mathcal{N}\right),
$
which links stability certification with the control efforts.

\subsection{Stability of Interconnected Systems }
In this section we derive Lyapunov functions from the storage functions.
{
We define the set $\mathcal{S}:=\left\{ \bigotimes_{i=1}^{n+m}\mathcal{S}_i \right\}$ to be the states that satisfy strictly EIP for each input-output pairs in all the subsystems. }
The next lemma allows us to construct Lyapunov functions for any equilibrium that is contained in $\mathcal{S}$. {Consequently, $\mathcal{S}$ is a subset of the states where the system will remain stable.} 

\begin{lemma}\label{lemma:Network_full_stable}
Consider the networked system~\eqref{eq:Modular_Line}-\eqref{eq:Matrix_M} with input $\bm{u}$ and output $\bm{y}$ that interconnected through $\bm{u}=\bm{M}\bm{y}$, where each input-output pair $\{u_i,y_i\}$ is locally strictly EIP with  $\epsilon_i$ for $\bm{\delta}\in\mathcal{S}$. If there exists a diagonal matrix $\bm{C} \succ 0$ such that $\bm{C}(\bm{M}-\hat{\bm{\epsilon}})+(\bm{M}-\hat{\bm{\epsilon}})^\top\bm{C}\prec 0$, then any equilibrium $\bm{\delta}^*\in\mathcal{S}$ is locally  asymptotically stable.
\end{lemma}
\begin{proof}
The proof roughly follows~\cite{arcak2016networks}. For completeness, we provide the key steps.
For the system~\eqref{eq:Modular_Line}-\eqref{eq:Matrix_M}, let the sum of the storage functions
$
V(\bm{\delta})= \sum_{i=1}^{n+2m} c_i V_{i}\left(\bm{\delta}\right) 
$
 serve as a candidate Lyapunov function. Its time derivative is
\begin{align}
    \dot{V}\left(\bm{\delta}\right) &=
 \sum_{i=1}^{n+2m}  c_i\dot{V}_{i}\left(\bm{\delta}\right)  \label{eq:Lya_sum} \\
    &\leq \sum_{i=1}^{n+2m}c_i
    \begin{bmatrix}
    u_{i}-u^*_i\\
    y_{i}-y^*_i
    \end{bmatrix}^{\top} 
     \begin{bmatrix}
    0 & 1 / 2\\
    1 / 2 & -\epsilon_{i}
    \end{bmatrix}
     \begin{bmatrix}
    u_{i}-u^*_i\\
    y_{i}-y^*_i
    \end{bmatrix} \nonumber \\
    &=\frac{1}{2} \begin{bmatrix}
 \bm{u}-\bm{u}^*\\
 \bm{y}-\bm{y}^*
\end{bmatrix}^{\top} 
    \begin{bmatrix}
     \bm{0} & \bm{C}\\
     \bm{C} &-2 \bm{C} \hat{\bm{\epsilon}}
    \end{bmatrix}
    \begin{bmatrix}
     \bm{u}-\bm{u}^*\\
     \bm{y}-\bm{y}^*
    \end{bmatrix} \nonumber\\
    &=\frac{1}{2}
    \begin{bmatrix}
 \bm{y}-\bm{y}^*
\end{bmatrix}^{\top}
    \begin{bmatrix}
    \bm{M} \\
    \bm{I}
    \end{bmatrix}^{\top}
    \begin{bmatrix}
     \bm{0} & \bm{C} \\
     \bm{C}  &-2 \bm{C} \hat{\bm{\epsilon}}
    \end{bmatrix}
    \begin{bmatrix}
    \bm{M} \\
    \bm{I}
    \end{bmatrix}
    \begin{bmatrix}
 \bm{y}-\bm{y}^*
\end{bmatrix}\nonumber\\
    &=\frac{1}{2}
    \left (\bm{y}-\bm{y}^*\right )^{\top}
    \left (\bm{C}(\bm{M}-\hat{\bm{\epsilon}})+(\bm{M}-\hat{\bm{\epsilon}})^\top\bm{C}\right )
     \left (\bm{y}-\bm{y}^*\right ) \nonumber
\end{align}

Because $\bm{y}=\bm{y}^*$ if and only if $\bm{\delta} = \bm{\delta}^*$,  $\bm{C}(\bm{M}-\hat{\bm{\epsilon}})+(\bm{M}-\hat{\bm{\epsilon}})^\top\bm{C}\prec 0$ implies $\dot{V}\left(\bm{\delta}\right)<0$ for $\bm{y}\neq \bm{y}^*$. Hence $V(\bm{\delta})$ is a valid Lyapunov function for $\bm{\delta}\in\mathcal{S}$, and an equilibrium $\bm{\delta}^*\in\mathcal{S}$ is locally asymptotically stable.
\end{proof}

The LMI in Lemma~\ref{lemma:Network_full_stable} is not jointly convex in $\bm{d}_{\mathcal{N}}$ or $\bm{C}$. The next theorem shows how the damping coefficients  $\bm{d}_\mathcal{N}$ can be designed based on the special structure of the interconnection matrix $\bm{M}$.


\begin{thm}[Local Exponential Stability]\label{thm:minimum_damp}
If the damping coefficients satisfy $ \operatorname{diag}(\bm{d}_\mathcal{N})\succ\frac{1}{4}\bm{\Psi }\hat{\bm{\epsilon}}_\mathcal{L}^{-1}\bm{\Psi }^\top$, an equilibrium $\bm{\delta}^*\in\mathcal{S}$ of the system ~\eqref{eq:dynamics}-\eqref{eq:power_injection} is locally exponentially stable.
\end{thm}
\begin{proof}
This theorem follows from picking $\bm{C}$ to be the identity matrix. In this case, the condition in Lemma~\ref{lemma:Network_full_stable} becomes $\left (\bm{M}^{\top}+\bm{M}-2 \hat{\bm{\epsilon}}\right )\prec 0$.  From~\eqref{eq:Matrix_M}, $\bm{M}=\bm{M_1}+\bm{M_2}$, and using the fact that $\bm{M_1}$ is skew symmetric, 
and expanding $\hat{\bm{\epsilon}} := \operatorname{diag}\left(\bm{d}_{\mathcal{N}}, \bm{\epsilon}_\mathcal{L}\right)$, we have
\begin{equation}
\dot{V}(\bm{\delta})=
\left (\bm{y}-\bm{y}^*\right )^\top
    \begin{bmatrix}
-\operatorname{diag}(\bm{d}_{\mathcal{N}})  & \frac{1}{2}\bm{\Psi } \\ \frac{1}{2}\bm{\Psi }    & -\hat{\bm{\epsilon}_\mathcal{L}}
\end{bmatrix}
\left (\bm{y}-\bm{y}^*\right ).
\end{equation}


To certify exponential stability, we need to find a scalar $\sigma>0$, such that 
$\dot{V}(\bm{\delta})<-\sigma V(\bm{\delta})$. Since the Lyapunov function $V(\bm{\delta})= \sum_{i=1}^{n+2m} c_i V_{i}\left(\bm{\delta}\right) 
$ is
\begin{equation*}
    \begin{aligned}
    V(\bm{\delta}) &=  \sum_{i=1}^n \frac{1}{2\tau_{ai}}\left (\delta_i-\delta^*_i\right )^2\\
    &= \left (\bm{y}-\bm{y}^*\right )^\top\begin{bmatrix}
\frac{1}{2}\operatorname{diag}(\bm{\tau})^{-1}  &  \bm{0}_{n\times 2m}\\ \bm{0}_{2m\times n}   & \bm{0}_{2m\times 2m}
\end{bmatrix} \left (\bm{y}-\bm{y}^*\right )
    \end{aligned},
\end{equation*}
then 
$\dot{V}(\bm{\delta})<-\sigma V(\bm{\delta})$ is equivalent to

\begin{equation}\label{eq:exp_stable}
\begin{split}
    \begin{bmatrix}
2\operatorname{diag}(\bm{d}_{\mathcal{N}})-\sigma\operatorname{diag}(\bm{\tau})^{-1}  & -\bm{\Psi } \\ -\bm{\Psi }    & 2\hat{\bm{\epsilon}_\mathcal{L}}
\end{bmatrix}\succ 0.
\end{split}
\end{equation}

By definition, $\hat{\bm{\epsilon}_\mathcal{L}}\succ 0$ and  Schur complement gives
\begin{equation*}
   \left( 2\operatorname{diag}(\bm{d}_{\mathcal{N}})-\sigma\operatorname{diag}(\bm{\tau})^{-1} \right)-\frac{1}{2}\bm{\Psi }\hat{\bm{\epsilon}_\mathcal{L}}^{-1}\bm{\Psi }^\top\succ 0.
\end{equation*}

If $ \operatorname{diag}(\bm{d}_\mathcal{N})
\succ
\frac{1}{4}\bm{\Psi }\hat{\bm{\epsilon}}_\mathcal{L}^{-1}\bm{\Psi }^\top$, then any $\sigma$ satisfying $0<\sigma<\lambda_{\min}\left(2\operatorname{diag}(\bm{d}_{\mathcal{N}})-\frac{1}{2}\bm{\Psi }\hat{\bm{\epsilon}}_\mathcal{L}^{-1}\bm{\Psi }^\top\right)\min_{i=1}^n\tau_{ai}$ guarantees~\eqref{eq:exp_stable} and therefore the equilibrium $\bm{\delta}^*$ is locally exponentially stable.
\end{proof}

Note that the damping coefficients obtained in Theorem~\ref{thm:minimum_damp} is derived by setting $\bm{C}=\bm{I}$, thus the region of stabilizing damping coefficients $ \operatorname{diag}(\bm{d}_\mathcal{N})\succeq\frac{1}{4}\bm{\Psi }\hat{\bm{\epsilon}}_\mathcal{L}^{-1}\bm{\Psi }^\top$ is a subset of that verified through $\bm{C}(\bm{M}-\hat{\bm{\epsilon}})+(\bm{M}-\hat{\bm{\epsilon}})^\top\bm{C}\prec 0$. We will show in the case study that the damping coefficients obtained by $ \operatorname{diag}(\bm{d}_\mathcal{N})\succeq\frac{1}{4}\bm{\Psi }\hat{\bm{\epsilon}}_\mathcal{L}^{-1}\bm{\Psi }^\top$ is already much less conservative compared with existing LMIs-based methods~\cite{zhang2016transient}.


%% file: EIP.tex
In this section, we prove the strictly EIP of the subsystems in~\eqref{eq:Modular_Line} and~\eqref{eq:Modular_Generator}. The system stability region is built from the angles that stabilize each of the subsystems. We also show how each stability region can be tuned to tradeoff with the size of the stabilizing damping coefficients.

\subsection{Strictly EIP of Lossy Transmission Lines and Buses}
The next Lemma shows that the subsystem~\eqref{eq:Modular_Line} of each lossy transmission line $l\in \mathcal{L}$ is strictly EIP for a region $\mathcal{S}_l$. 

\begin{lemma}
[EIP of Lossy Lines]\label{lemma:EIP_line}
The lossy transmission line $l$ from bus $i$ to $j$ represented by~\eqref{eq:Modular_Line} is strictly EIP with $\epsilon_l = \frac{2\alpha_l}{\sqrt{g_l^2+b_l^2\alpha_l^2}}$ for all the possible equilibriums $\delta_{ij}^*$ in the set $\mathcal{S}_l=\{\delta_{ij}^*| -\arctan{(b_l\alpha_l/g_l)}\leq \delta_{ij}^* \leq \arctan{(b_l\alpha_l/g_l)}\}$. 
\end{lemma}

First we note that if $g_l=0$, then the subsystem~\eqref{eq:Modular_Line} is strictly EIP in $\delta_{ij}^*\in(-\frac{\pi}{2}, \frac{\pi}{2})$ for any $\alpha_l>0$ . In particular, $\bm{\Psi}$ can be made arbitrarily close to $\bm{0}$ and  $ \operatorname{diag}(\bm{d}_\mathcal{N})\succ\frac{1}{4}\bm{\Psi }\hat{\bm{\epsilon}}_\mathcal{L}^{-1}\bm{\Psi }^\top$ for any $\bm{d}_\mathcal{N} > 0$. Namely, $\delta_{ij}^*\in(-\frac{\pi}{2}, \frac{\pi}{2})$ is stable for any positive damping coefficients. This recovers the observations for lossless transmission lines~\cite{cui2021reinforcement}. 

For lossy transmission line with $g_l>0$, Lemma~\ref{lemma:EIP_line} shows that $\alpha_l$ trades off between the size of $\mathcal{S}_l$ and passivity: a larger $\alpha_l$ enlarges $\mathcal{S}_l$ but also
increases the bound
$\frac{1}{4}\bm{\Psi }\hat{\bm{\epsilon}}_\mathcal{L}^{-1}\bm{\Psi }^\top$ that requires larger damping. 
The proof is given below.

\begin{proof}
The subsystem~\eqref{eq:Modular_Line} is a memoryless, where $y_{l,1}$ and $y_{l,2}$ is a function of the input $u_{l,1}=\delta_{ij}$ and $u_{l,2}=-\delta_{ij}$,   respectively. Hence, it is suffices to consider the function
\begin{equation}\label{eq:abstract_y_u}
\begin{aligned}
y_l(u) =& \frac{g_l-g_l\cos(u) }{2\alpha_l}+\frac{b_l}{2}\sin(u)\\
=&\frac{g_l}{2\alpha_l}+\frac{\sqrt{g_l^2+b_l^2\alpha_l^2}}{2\alpha_l}\sin(u-\gamma_{l}), 
\end{aligned}
\end{equation}
when $u=\delta_{ij}$ and $u=-\delta_{ij}$, respectively. The constant $\gamma_l=\arctan(\frac{g_l}{b_l\alpha_l})\in(0,\pi/2)$ horizontally shift the function $y_l(u)$ as shown in Fig.~\ref{fig:sin_EIP} and thus affect the range of $\delta_{ij}$ satisfying strictly EIP.
For the memoryless system~\eqref{eq:abstract_y_u}, we take the storage function to be zero
and then the condition for strict passivity is~\cite{arcak2016networks}
\begin{equation}\label{eq:memoryless_EIP}
\left(u-u^*\right)\left(y_l(u)-y_l(u^*)\right)-\epsilon_l
\left(y_l(u)-y_l(u^*)\right)^2\geq 0,
\end{equation}
which holds for any equilibrium  if and only if  $y_l^\prime(u)\in \left [0, \frac{1}{\epsilon_l}  \right ]$ (detailed proof is given in Appendix~\ref{app:memoryless}. To this end, setting $\epsilon_l = \frac{2\alpha}{\sqrt{g^2+b^2\alpha^2}}$ guarantees that $y_l'(u)\leq\frac{1}{\epsilon_l}$. Then  $y_l'(u)\geq 0$ is guaranteed for the region $u\in [-\frac{\pi}{2}+\gamma_l, \frac{\pi}{2}+\gamma_l]$, which is labeled in red  in Fig.~\ref{fig:sin_EIP}.

Substituting $u=\delta_{ij}$ and $u=-\delta_{ij}$ gives $-\frac{\pi}{2}\leq \delta_{ij}-\gamma_l\leq \frac{\pi}{2}$ and  $-\frac{\pi}{2}\leq -\delta_{ij}-\gamma_l\leq \frac{\pi}{2}$, respectively. Taking the intersection, the angle difference satisfying strictly EIP is $\delta_{ij}\in [-\frac{\pi}{2}+\gamma_l, \frac{\pi}{2}-\gamma_l ]$, which is equivalent to $\delta_{ij}\in [-\arctan{(b_l\alpha_l/g_l)}, \arctan{(b_l\alpha_l/g_l)}].$
\end{proof}

\begin{figure}[ht]	
	\centering
	\includegraphics[width=3.5in]{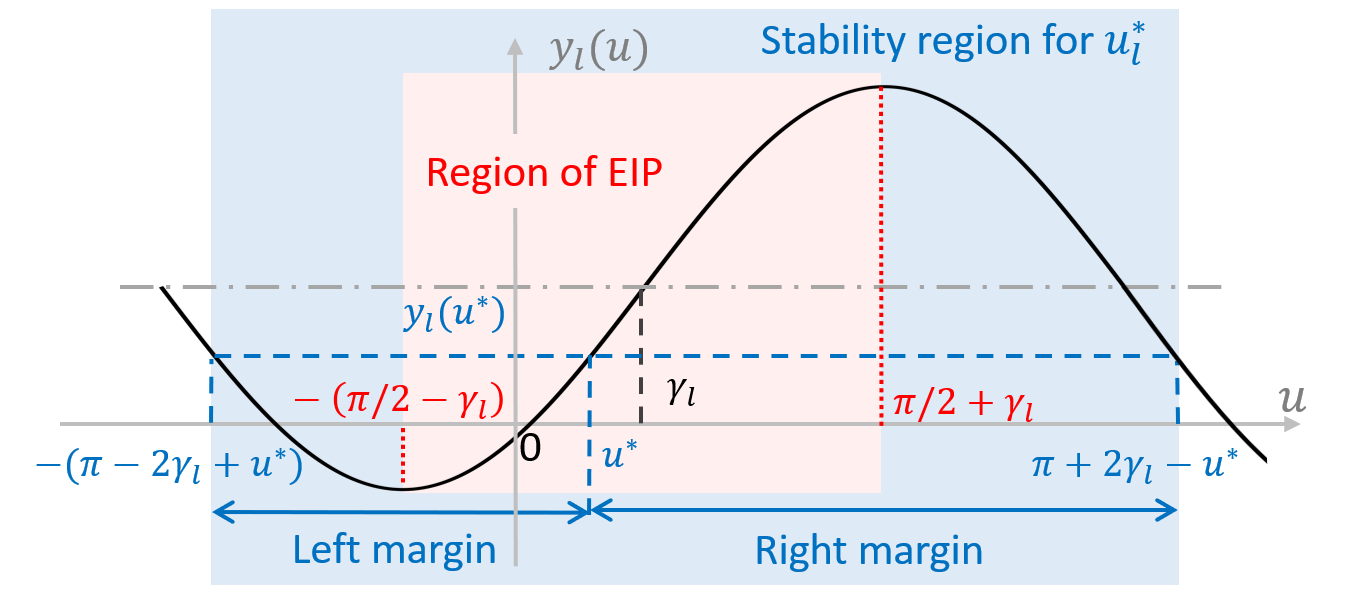}
	\caption{ The region of EIP $\mathcal{S}_l$ is computed by $y_l'(u)\geq 0$ and is labeled in red. The stability region for an equilibrium $u^*$ is the areas that  $y_l(u)-y_l(u^*)$ and $u-u^*$ has the same sign. The stability region is labeled in blue, and its intersection for all equilibrium $u^*\in\mathcal{S}_l$ is the region of EIP in red. \vspace{-0.5cm}}
	\label{fig:sin_EIP}
\end{figure}

\begin{lemma}
[EIP of buses]\label{theorem:EID_bus}
Bus $k$ represented by~\eqref{eq:Modular_Generator} is strictly EIP with $\epsilon_k =d_{ak}$ for all equilibria  $\delta^*_k\in \mathbb{R}$.
\end{lemma}
This Lemma shows that the subsystem of buses is strictly EIP for all the possible equilibrium of angles. It follows directly from the definitions and we omit the proof. 


\subsection{Sizing Stability Regions}
The equilibrium-independent stability guarantees that any equilibrium in the set $\mathcal{S}$ is exponentially stable. Naturally, it is of interest to control the size of the stability region $\mathcal{S}$ (sometimes called region of attraction). 
The next theorem shows how the parameter $\alpha$ should be chosen if the  stability region  need to meet a prescribed size.

\begin{thm}
[Tuning $\alpha$ for Stability Region]\label{theorem:Stablizing_region}
For the line $l$ from bus $i$ to $j$ with an equilibrium $\delta_{ij}^*\in\mathcal{S}_l$, the stability region is $\mathcal{S}_{l}|_{\bm{\delta}^*}=\{\delta_{ij}| -2\arctan{(b_l\alpha_l/g_l)}-\delta_{ij}^*\leq \delta_{ij} \leq 2\arctan{(b_l\alpha_l/g_l)}-\delta_{ij}^*\}$.
If $\alpha_l\geq\frac{g_l\tan(|\delta_{ij}^*|+\beta_l/2)}{b_l}$ for a constant $0<\beta_l<\pi-2|\delta_{ij}^*|$, then the system is guaranteed to be stable around the equilibrium $\delta_{ij}^*$ with at least the margin of $\beta_l$, i.e., $[\delta_{ij}^*-\beta_l, \delta_{ij}^*+\beta_l]\subset \mathcal{S}_{l}|_{\bm{\delta}^*}$. 
\end{thm}
Note that if varying $\delta_{ij}^*$ in the set $\mathcal{S}_l=\{\delta_{ij}^*| -\arctan{(b_l\alpha_l/g_l)}\leq \delta_{ij}^* \leq \arctan{(b_l\alpha_l/g_l)}\}$, the intersection of $\mathcal{S}_{l}|_{\bm{\delta}^*}$ is exactly  $\mathcal{S}_l$. Hence, the region of equilibrium-independent stability can also be understand as the intersection of the stability region for all the possible equilibrium.


\begin{proof}
The stability certification~\eqref{eq:Lya_sum}-\eqref{eq:exp_stable} holds as long as the inequality~\eqref{eq:def_sEIP} holds. For a certain equilibrium $\bm{\delta}^*$, we define the  stability region $\mathcal{S}_{l}|_{\bm{\delta}^*}$ to be the angles satisfying the inequality~\eqref{eq:def_sEIP}. This condition is equivalent to certifying~\eqref{eq:memoryless_EIP} for  $u=\delta_{ij}$ and $u=-\delta_{ij}$ when fixing $u^*=\delta_{ij}^*$.  Note that $\epsilon_l = \frac{2\alpha}{\sqrt{g^2+b^2\alpha^2}}$ gives $y_l'(u)\leq\frac{1}{\epsilon_l}$, then condition~\eqref{eq:memoryless_EIP} is satisfied as long as $y_l(u)-y_l(u^*)$ is the same sign as $u-u^*$ for both $u=\delta_{ij}$ and $u=-\delta_{ij}$.




The signs of $y_l(u)-y_l(u^*)$ and $u-u^*$ are the same when $u\in [-\pi+2\gamma_l-u^*,  \pi+2\gamma_l-u^*]
$. This region is labeled in blue in Fig.~\ref{fig:sin_EIP}, which is larger than the region of EIP shown in red.
For $u=\delta_{ij}$ and $u=-\delta_{ij}$, we have
$
\delta_{ij}\in [-\pi+2\gamma_l-\delta^*_{ij},  \pi+2\gamma_l-\delta^*_{ij}],
$
and
$
-\delta_{ij}\in [-\pi+2\gamma_l+\delta^*_{ij},  \pi+2\gamma_l+\delta^*_{ij}]
$, respectively. 
The intersection gives the  region
\begin{equation}\label{eq:stable_angle_diff}
\mathcal{S}_{l}|_{\bm{\delta}^*}=\{\delta_{ij}| -\pi+2\gamma_l-\delta^*_{ij}\leq \delta_{ij} \leq \pi-2\gamma_l-\delta^*_{ij}\}.
\end{equation}
and thus $[\delta_{ij}^*-\beta_l, \delta_{ij}^*+\beta_l]\subset \mathcal{S}_{l}|_{\bm{\delta}^*}$ yields
$$
 -\pi+2\gamma_l-\delta^*_{ij}
 \leq 
 \delta^*_{ij}-\beta_l
 \leq 
 \delta^*_{ij}+\beta_l
 \leq 
 \pi-2\gamma_l-\delta^*_{ij},
$$
which gives $\frac{\pi}{2}-\gamma_l\geq|\delta^*_{ij}|+\frac{\beta_l}{2}$. Equivalently, $\arctan(\frac{b_l\alpha}{g_l})\geq|\delta^*_{ij}|+\frac{\beta_l}{2}$ and thus we require $\alpha_l\geq\frac{g_l\tan(|\delta^*_{ij}|+\beta_l/2)}{b_l}$.
\end{proof}

Theorems~\ref{thm:minimum_damp} and~\ref{theorem:Stablizing_region} provide a way of optimizing over the damping coefficients while guaranteeing the size of the stability region. Specifically, suppose the margin of stable angle difference is $\beta_l\in[0,\pi-2|\delta^*_{ij}|]$ for $l\in  \mathcal{L}$, then we define $\alpha_l=\frac{g_l\tan(|\delta^*_{ij}|+\beta_l/2)}{b_l}$. Thus, $\epsilon_l = \frac{2\alpha_l}{\sqrt{g_l^2+b_l^2\alpha_l^2}}$ and the matrix $\bm{\Psi}$ is determined by $\alpha_l$'s through \eqref{subeq:Modular_Generator_y}. 
To minimize the damping coefficients (corresponding to hardware costs~\cite{johnson2015synthesizing}), we can solve 
\begin{subequations}\label{eq:Optimization}
\begin{align}
\min_{\bm{d}_\mathcal{N}} & \quad \|\bm{d}_\mathcal{N}\|_2\label{subeq:Optimization_obj}\\
\mbox{s.t. } & \operatorname{diag}(\bm{d}_\mathcal{N})\succ\frac{1}{4}\bm{\Psi }\hat{\bm{\epsilon}}_\mathcal{L}^{-1}\bm{\Psi }^\top\label{subeq:Optimization_constraint},
\end{align}
\end{subequations}
which is a convex problem. 
The Pareto-front of the least-cost damping coefficients and the size of stability region can be computed by varying $\alpha$, quantifying 
the trade-off between control efforts and stability regions.

\subsection{Algorithm}
We illustrate the implementation of the proposed technique in the following algorithm. 

\begin{algorithm}
 \caption{Equilibrium-Independent  Stability Analysis and Control Effort Optimization
 }
 \begin{algorithmic}[1]
 \renewcommand{\algorithmicrequire}{\textbf{Require: }}
 \renewcommand{\algorithmicensure}{\textbf{Initialisation:}}  
 \REQUIRE  The graph  $\left(\mathcal{N},\mathcal{L} \right)$, the conductance $g_l$, susceptance $b_l$ and the tunable parameter $\alpha_l$ of the transmission line $l\in\mathcal{L}$\\
 \ENSURE $\bm{\Psi}\in\bm{0}_{n\times 2m}$, $\bm{\epsilon}_{\mathcal{L}}\in\bm{0}_{2m}$\\
  \FOR {$l = n+1$ to $n+m$}
  \STATE If $l=(i,j)$, $[\bm{\Psi}]_{i,2(l-n)-1}=-\alpha_l$, $[\bm{\Psi}]_{j,2(l-n)}=-\alpha_l$
  . \\
  \STATE Set $[\bm{\epsilon}_{\mathcal{L}}]_{2(l-n)-1 }= \frac{2\alpha_l}{\sqrt{g_l^2+b_l^2\alpha_l^2}}$, $[\bm{\epsilon}_{\mathcal{L}}]_{2(l-n)}= \frac{2\alpha_l}{\sqrt{g_l^2+b_l^2\alpha_l^2}}$\\
  \ENDFOR\\
 \renewcommand{\algorithmicensure}{\textbf{Stability Verification:}}  
 \ENSURE  Define $\hat{\bm{\epsilon}}_\mathcal{L} = \operatorname{diag}(\bm{\epsilon}_\mathcal{L})$. Input $\bm{d}_\mathcal{N}:=\left(d_{a1},\cdots, d_{an}\right)$\\
  \IF{$ \operatorname{diag}(\bm{d}_\mathcal{N})\succeq\frac{1}{4}\bm{\Psi }\hat{\bm{\epsilon}}_\mathcal{L}^{-1}\bm{\Psi }^\top$}
    \STATE Stable
  \ELSE
    \STATE Unstable
  \ENDIF
 \renewcommand{\algorithmicensure}{\textbf{Control Effort Optimization:}}  
 \ENSURE Define $\hat{\bm{\epsilon}}_\mathcal{L} = \operatorname{diag}(\bm{\epsilon}_\mathcal{L})$\\
 \STATE Run SDP solver for $\min_{\bm{d}_\mathcal{N}}   \|\bm{d}_\mathcal{N}\|_2 \quad \mbox{s.t. }  \operatorname{diag}(\bm{d}_\mathcal{N})\succ\frac{1}{4}\bm{\Psi }\hat{\bm{\epsilon}}_\mathcal{L}^{-1}\bm{\Psi }^\top$
 \end{algorithmic} 
 \end{algorithm}




%% file: simulation.tex
Case studies are conducted on the IEEE 123-node test feeder~\cite{kersting1991radial}. Since existing LMIs-based and neural network-based stability assessment methods all partition the network into a 5-bus system to alleviate computational issues~\cite{zhang2016transient, huang2021neural}, we first work with this 5-bus system as well to show that the proposed method can achieve larger stability regions with smaller damping coefficients. Then, we directly work with the original 123-node feeder to show that the proposed approach can scale to large systems.

\subsection{Comparison with LMIs-Based Stability Assessment }
We first compare with existing LMI-based transient stability assessment found in~\cite{zhang2016transient}. The paramter of the test system (partitioned into 5 buses) can be found in~\cite{zhang2016transient, huang2021neural}. 


Under the same damping coefficients $\bm{d}_\mathcal{N}=[0.27, 0.22, 0.27, 0.21, 0.21]$, Fig.~\ref{fig:Stablizing_Region} compares the stability region of two lines calculated by our proposed method and the benchmark LMIs-based method in~\cite{zhang2016transient}. The angle difference $\delta_{ij}$ relative to an equilibrium for the line connecting bus $i$ and $j$ are labeled as $\Delta\delta_{ij}:=\delta_{ij}-\delta^*_{ij}$. 
Our proposed approach attains much larger stability region.


From the other direction, if we fix the size of the stability regions, \eqref{eq:Optimization} can be solved to find the stabilizing damping coefficients. This is in contrast to existing methods, where damping coefficients are found through  exhaustive searches. 


\begin{figure}[ht]
\centering
\includegraphics[width=3.in]{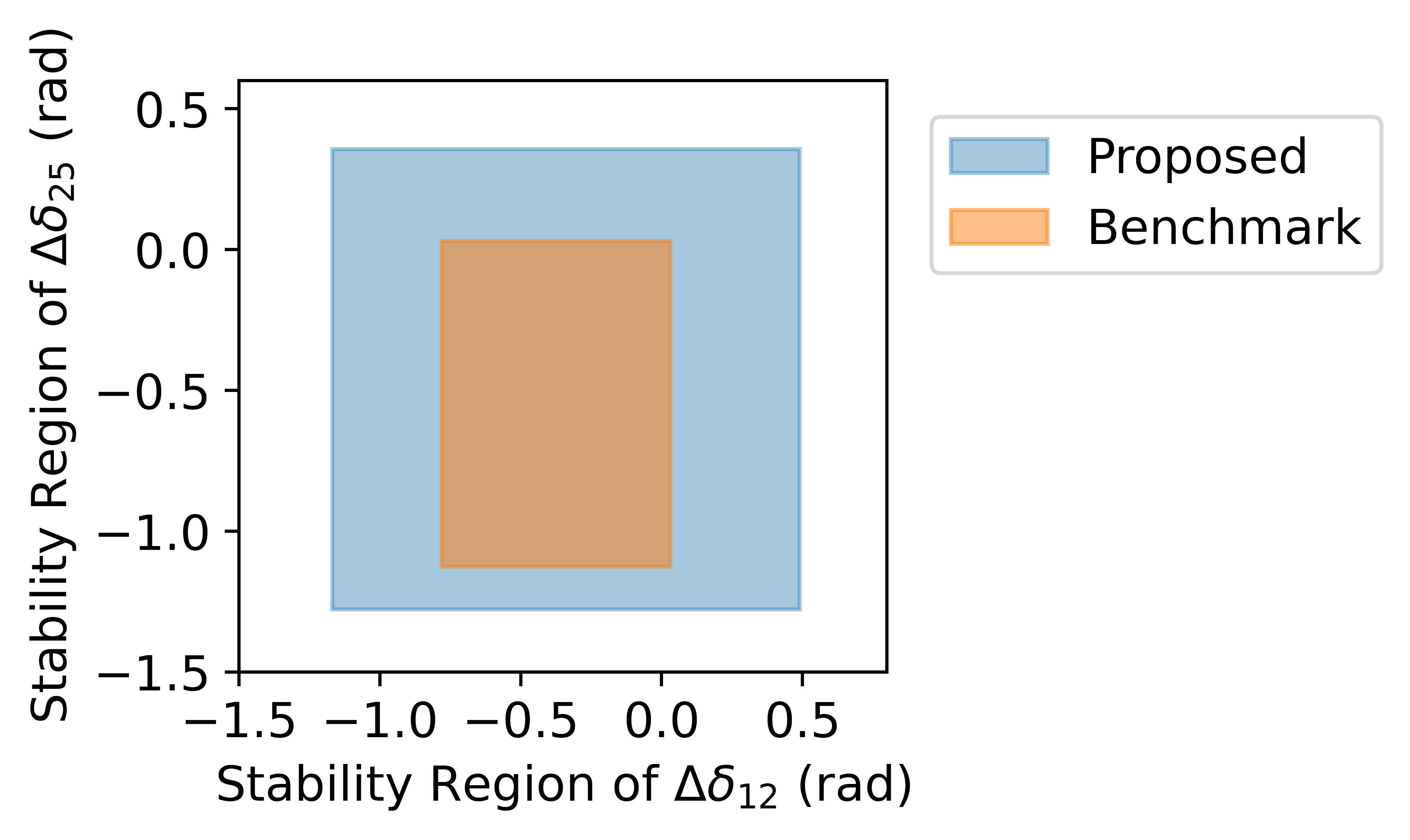}
\caption{Stability regions of $\Delta\delta_{12}$ and $\Delta\delta_{25}$ under the same damping coefficients. The proposed approach finds a larger
stability region. \vspace{-0.5cm} }
\label{fig:Stablizing_Region}
\end{figure}




\subsection{Performance on Large Systems}

To verify the performance of the proposed method on larger systems, we further simulate on the original 123-node test feeder. Fig.~\ref{fig:Dynamics} compares the dynamics of the system with different damping coefficients.
The system stabilizes to the setpoints in the former and diverging in the latter case.  
Moreover, Fig.~\ref{fig:Balance_alpha} shows the Pareto-front of the width of the stability region 
and the least-norm stabilizing damping coefficient by varying $\alpha$ from 0.1 to 2 in the line 1. This quantifies the trade-off between enlarging the stability region and minimizing control efforts. 

\begin{figure}[ht]
\centering
\subfloat[$\operatorname{diag}(\bm{d}_\mathcal{N})\succ\frac{1}{4}\bm{\Psi }\hat{\bm{\epsilon}}_\mathcal{L}^{-1}\bm{\Psi }^\top$]{\includegraphics[width=1.7in]{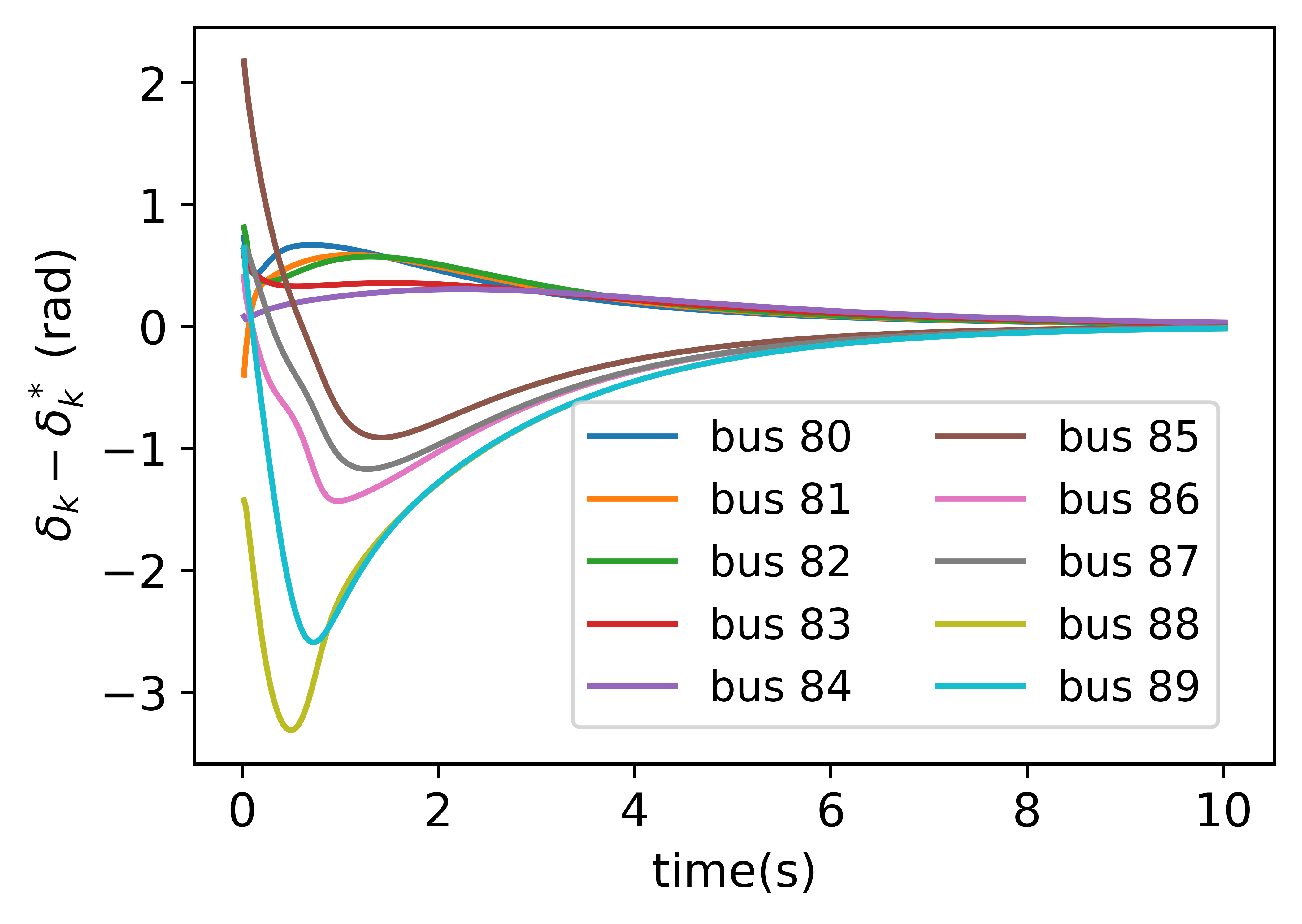}%
\label{fig_first_case}}
\subfloat[$\operatorname{diag}(\bm{d}_\mathcal{N})\prec \frac{1}{4}\bm{\Psi }\hat{\bm{\epsilon}}_\mathcal{L}^{-1}\bm{\Psi }^\top$]{\includegraphics[width=1.7in]{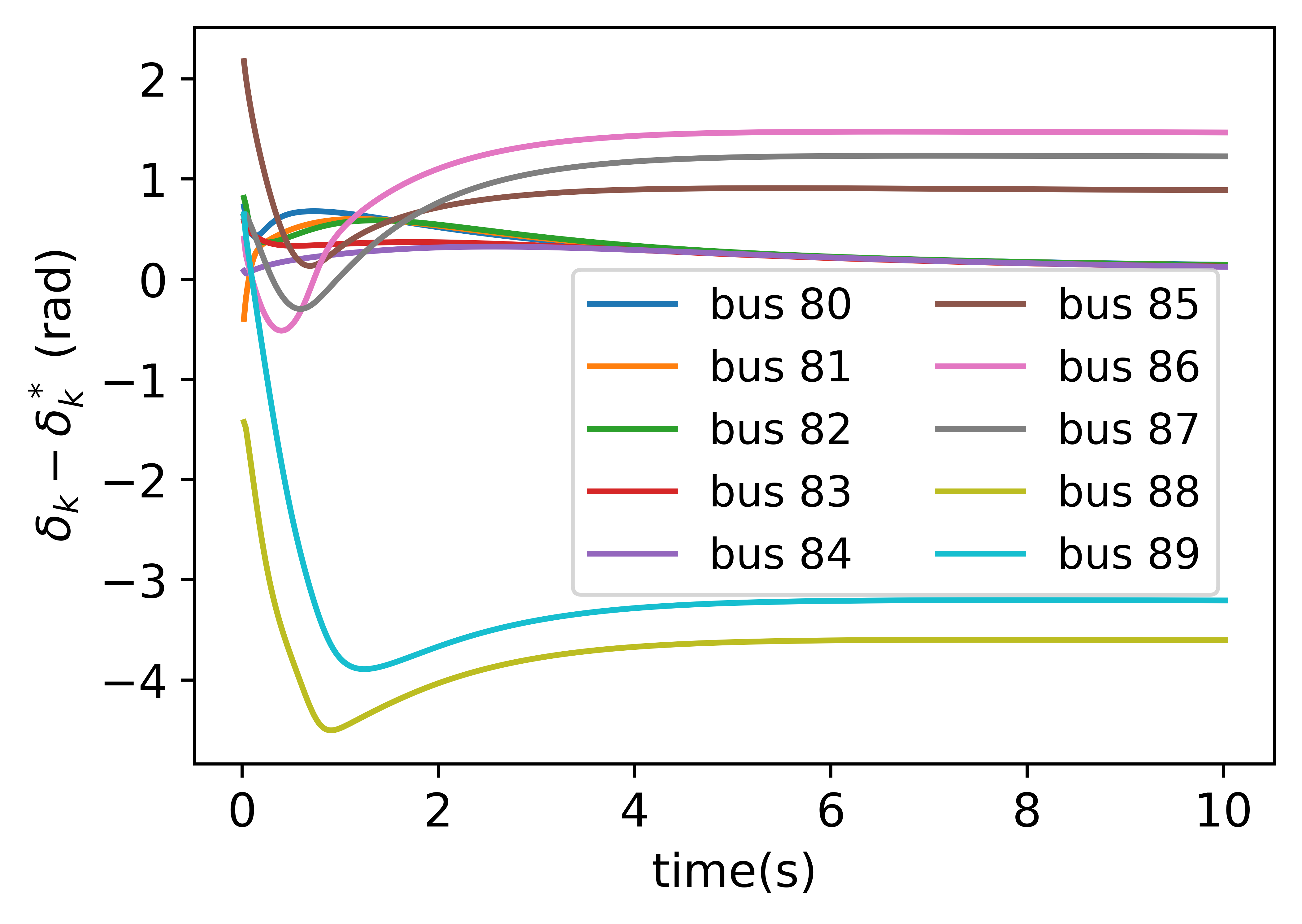}%
\label{fig_second_case}}
\caption{Dynamics of ten selected buses with (a) the damping coefficients satisfying the proposed bound in~\eqref{subeq:Optimization_constraint} (stable)  and (b) the reduced damping coefficients violate the proposed bound (diverges from setpoints). \vspace{-0.5cm} }
\label{fig:Dynamics}
\end{figure}


\begin{figure}[ht]	
	\centering
	\includegraphics[width=3.in
	]{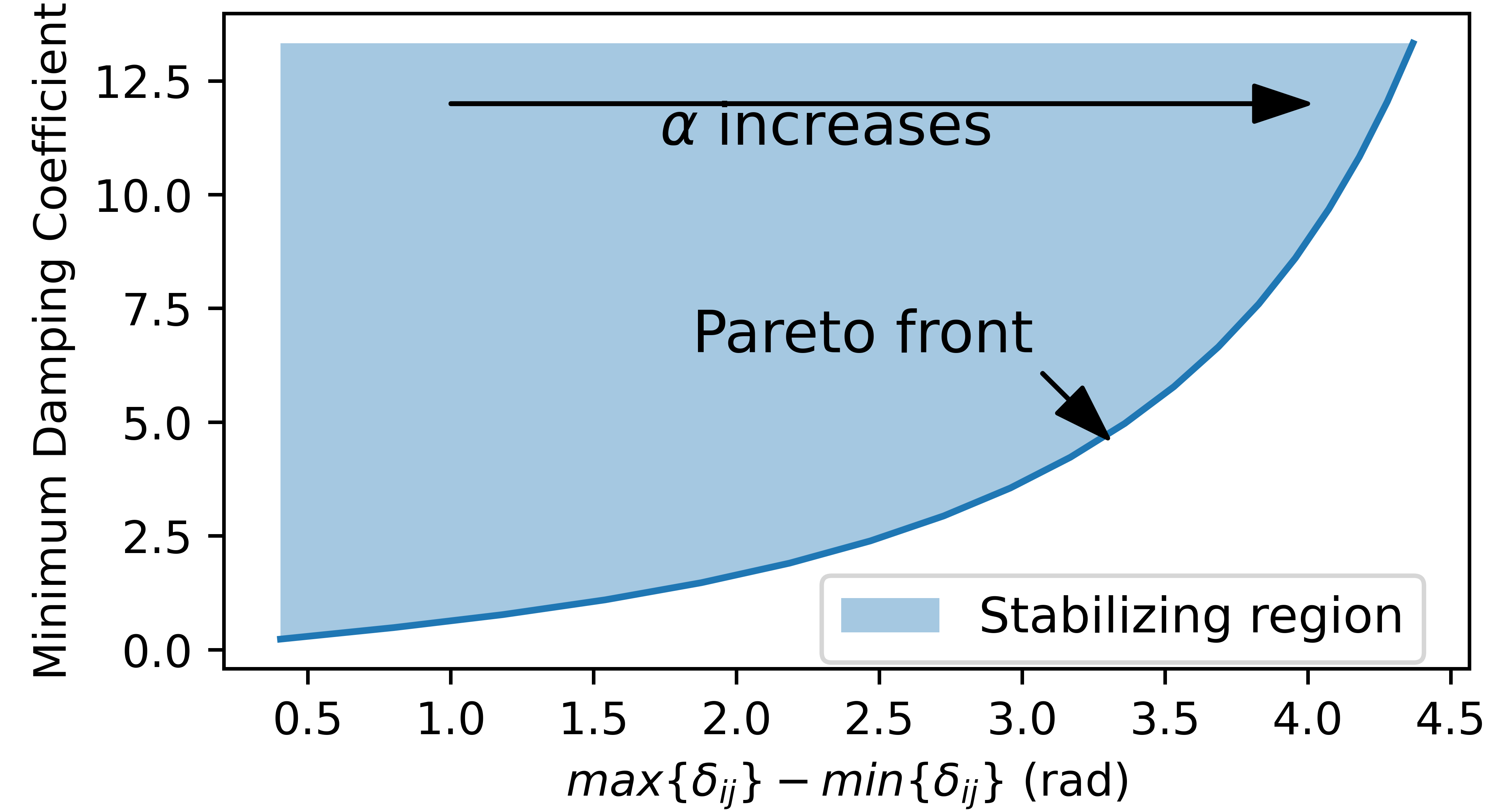}
	\caption{Pareto-front of the width of the stability region and the minimum stabilizing damping coefficients  by varying $\alpha$ from 0.1 to 2 in line 1. \vspace{-1.4cm}}
	\label{fig:Balance_alpha}
\end{figure}


%% file: appendix.tex
\subsection{Detailed Descriptions on the Modular Design}\label{app: modular}
 Let $l=(i,j)\in\mathcal{L}$ denotes the line connecting bus $i$ and $j$ with the flow from $i$ to $j$ be the positive direction. 
Compactly, the closed-loop dynamics of the bus $k\in\mathcal{N}$  is 
\begin{align}\label{eq:dyn_compact}
\tau_{\mathrm{a} k} \dot{\delta}_{k} &=
-d_{ak}(\delta_{k}-\delta_{k}^\ast)+ p_{k}^{*} \\
& -\sum_{\substack{l=(k,j) \text{ or }\\ l=(j,k)\in\mathcal{L}}}\underbrace{\left(g_l-g_l\cos(\delta_k-\delta_j)+b_l\sin(\delta_k-\delta_j)\right)}_{p_{kj}}, \nonumber
\end{align}
which is the ordinary differential equation  with power flow in the summation.

We would like to clarify that we define the modules for convenience of analysis, but it does not necessarily has physical counterpart. The inputs and outputs can be designed flexibly according to the application requirements, as long as the closed-loop dynamics of all the modules recover original system dynamics in~\eqref{eq:dyn_compact}. 
For example, the power flow in the line connect bus $i$ and $j$ is 
\begin{equation}\label{eq:dynamics_line}
\begin{split}
    p_{ij}=g_l-g_l\cos(\delta_i-\delta_j)+b_l\sin(\delta_i-\delta_j)\\
    p_{ji}=g_l-g_l\cos(\delta_j-\delta_i)+b_l\sin(\delta_j-\delta_i)
\end{split}
\end{equation}

We want the analysis to extrapolate easily between lossless and lossy lines, and therefore we introduce $\alpha_l$ on the terms related to the lossy lines. Next, we will show the exact definition of the inputs and outputs and how the closed-loop dynamics of all the modulars recover original system dynamics in~\eqref{eq:dyn_compact}.
For the line $l$ connects bus $i$ and $j$, we define the input as
\begin{equation}
    \begin{split}
    u_{l,1}&=\delta_i-\delta_j\\ 
    u_{l,2}&=\delta_j-\delta_i,
    \end{split}
\end{equation}
and the output as 
\begin{equation}
\begin{split}
    y_{l,1}&=\frac{1}{2\alpha_l}\left (g_l-g_l\cos(\delta_i-\delta_j) \right )+\frac{b_l}{2} \sin(\delta_i-\delta_j)\\
    y_{l,2}&=\frac{1}{2\alpha_l}\left (g_l-g_l\cos(\delta_j-\delta_i) \right )+\frac{b_l}{2} \sin(\delta_j-\delta_i)
\end{split}
\end{equation}
Using the fact that $\sin(\delta_j-\delta_i)=-\sin(\delta_i-\delta_j)$ and $\cos(\delta_j-\delta_i)=\cos(\delta_i-\delta_j)$, we have $y_{l,1}+y_{l,2}=(g_l-g_l\cos(\delta_i-\delta_j)/\alpha_l $ and $y_{l,1}-y_{l,2}=b_l \sin(\delta_i-\delta_j) $. Hence, the power flow~\eqref{eq:dynamics_line} can be recovered by
\begin{equation}\label{eq:power_flow_pij}
    \begin{split}
    p_{ij} &= \alpha_l(y_{l,1}+y_{l,2})+(y_{l,1}-y_{l,2})\\
    p_{ji} &= \alpha_l(y_{l,1}+y_{l,2})-(y_{l,1}-y_{l,2})
    \end{split}
\end{equation}
In a vector format, we have
\begin{equation}\label{eq:vector_power_flow}
    \begin{split}
    p_{ij} &=  \left[\begin{matrix}
     1 & -1
    \end{matrix}\right ] \bm{y}_l
    +\alpha_l\left[\begin{matrix}
     1 & 1
    \end{matrix}\right ] \bm{y}_l\\
    p_{ji} &= \left[\begin{matrix}
     -1 & 1
    \end{matrix}\right ] \bm{y}_l
    +\alpha_l\left[\begin{matrix}
     1 & 1
    \end{matrix}\right ] \bm{y}_l
    \end{split}
\end{equation}

The input of each node is the summation of the power injection from all the connected lines. This gives the representation in~\eqref{subeq:input_node} where the input $u_k$ is the summation of all the injected power flow as shown in~\eqref{eq:vector_power_flow}. Hence, the original dynamics~\eqref{eq:dyn_compact} are recovered by the following definition of modules.
\begin{subequations}\label{eq:Modular_Generator_4}
\begin{align}
\tau_{ k} \dot{\delta}_{k}& =-d_k(\delta_{k}-\delta_{k}^\ast)+ \left(P_{k}^{*}+u_{k}\right) \\
y_{ k} &=\delta_{k}\label{subeq:Modular_Generator_y_4}\\
  \begin{split}
      u_{ k} &=  \sum_{l\in  \mathcal{B}_k^+}
      \underbrace{\left[\begin{matrix}
     -1 & 1
    \end{matrix}\right ]}_{\Phi_{kl}} \bm{y}_l
    +\underbrace{\alpha_l\left[\begin{matrix}
     -1 & -1
    \end{matrix}\right ]}_{\Psi_{kl}} \bm{y}_l\\
    &\quad
    +\sum_{l\in  \mathcal{B}_k^-}
    \underbrace{\left[\begin{matrix}
     1 & -1
    \end{matrix}\right ]}_{\Phi_{kl}} \bm{y}_l
    +\underbrace{\alpha_l\left[\begin{matrix}
     -1 & -1
    \end{matrix}\right ]}_{\Psi_{kl}} \bm{y}_l\label{subeq:input_node}
    \end{split}
\end{align}
\end{subequations}

Due to the page limit, we supplement these detailed analysis in the longer online version~\cite{cui2022equilibrium}.

\subsection{Examples illustrating the structural properties of the proposed design}\label{app:example}

The following two examples provide more details on how the proposed method works. \\
\textbf{Two-bus network.} The first example is a radial network, where two buses are connected by a line as shown in Figure~\ref{fig:toy_example_bus2}. 
\begin{figure}[H]
    \centering
    \includegraphics[scale=0.5]{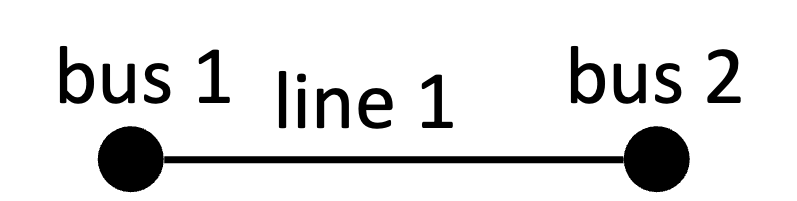}
    \caption{A two-bus (radial) network. The indices are $k=\{1,2\}$ for the buses and $l=\{3\}$ for the lines.}
    \label{fig:toy_example_bus2}
\end{figure}

We specify the input and output of this system by definition of modules given in~\eqref{eq:Modular_Generator} and~\eqref{eq:Modular_Line}. For each node, the input
$u_1=-p_{12}$, $u_2=-p_{21}$,and the output $y_1=\delta_1$, $y_2=\delta_2$.
For line 1, define power flow from bus 1 to bus 2 as the positive direction. We associate an index $l=3$ for this line. Then the input 
\begin{equation}
\begin{split}
    &\bm{u}_3=\begin{bmatrix}   \delta_1-\delta_2&  \delta_2-\delta_1 \end{bmatrix}^\top =\underbrace{\left[\begin{matrix}
     1 \\-1
    \end{matrix}\right ]}_{\Phi_{31}} \delta_1
    +
    \underbrace{\left[\begin{matrix}
     -1 \\1
    \end{matrix}\right ]}_{\Phi_{32}} \delta_2
\end{split}    
\end{equation}

and the output
\begin{equation}
    \begin{split}
    \bm{y}_3&=\begin{bmatrix}   \frac{\left (g_3-g_3\cos(\delta_1-\delta_2) \right )}{2\alpha_3}+\frac{b_3\sin(\delta_1-\delta_2)}{2}
    \\
    \frac{\left (g_3-g_3\cos(\delta_2-\delta_1) \right )}{2\alpha_3}+\frac{b_3\sin(\delta_2-\delta_1)}{2} \end{bmatrix}
    \end{split}
\end{equation}
From the power flow in the line $l$ connecting bus $i$ and $j$ given by
\begin{equation}\label{eq: power_flow_toy} 
p_{ij}=g_l-g_l\cos(\delta_i-\delta_j)+b_l\sin(\delta_i-\delta_j),    
\end{equation}
we have
\begin{equation}
    \begin{split}
   &     -p_{12}=  \underbrace{\left[\begin{matrix}
     -1 & 1
    \end{matrix}\right ]}_{\Phi_{13}} \bm{y}_3
    +\underbrace{\alpha_3\left[\begin{matrix}
     -1 & -1
    \end{matrix}\right ]}_{\Psi_{13}} \bm{y}_3,\\
    &-p_{21}= \underbrace{\left[\begin{matrix}
     1 & -1
    \end{matrix}\right ]}_{\Phi_{23}} \bm{y}_3
    +\underbrace{\alpha_3\left[\begin{matrix}
     -1 & -1
    \end{matrix}\right ]}_{\Psi_{23}} \bm{y}_3. 
    \end{split}
\end{equation}
To investigate the stability of the whole interconnected system, we stack the input/output vectors in sequence as $\bm{u}:=(u_1, u_2, \bm{u}_3^\top)$ and $\bm{y}:=(y_1, y_2, \bm{y}_3^\top)$. Their interconnection is therefore
\begin{equation*}
\includegraphics[width=0.85\linewidth]{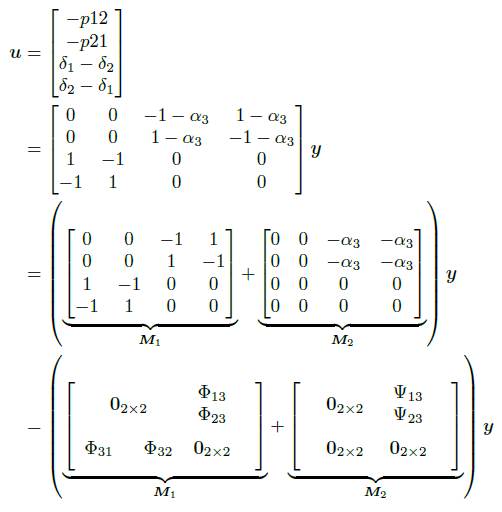}
\end{equation*}

\noindent \textbf{Three-bus meshed network.} Here we show an example of a meshed three-bus network in Fig.~\ref{fig:toy_example_bus3}. It is also included in the online version of the revised manuscript~\cite{cui2022equilibrium}. Nothing substantially changes from the two-bus network, except that the sizes of matrices get bigger. 
\begin{figure}[H]
    \centering
    \includegraphics[scale=0.5]{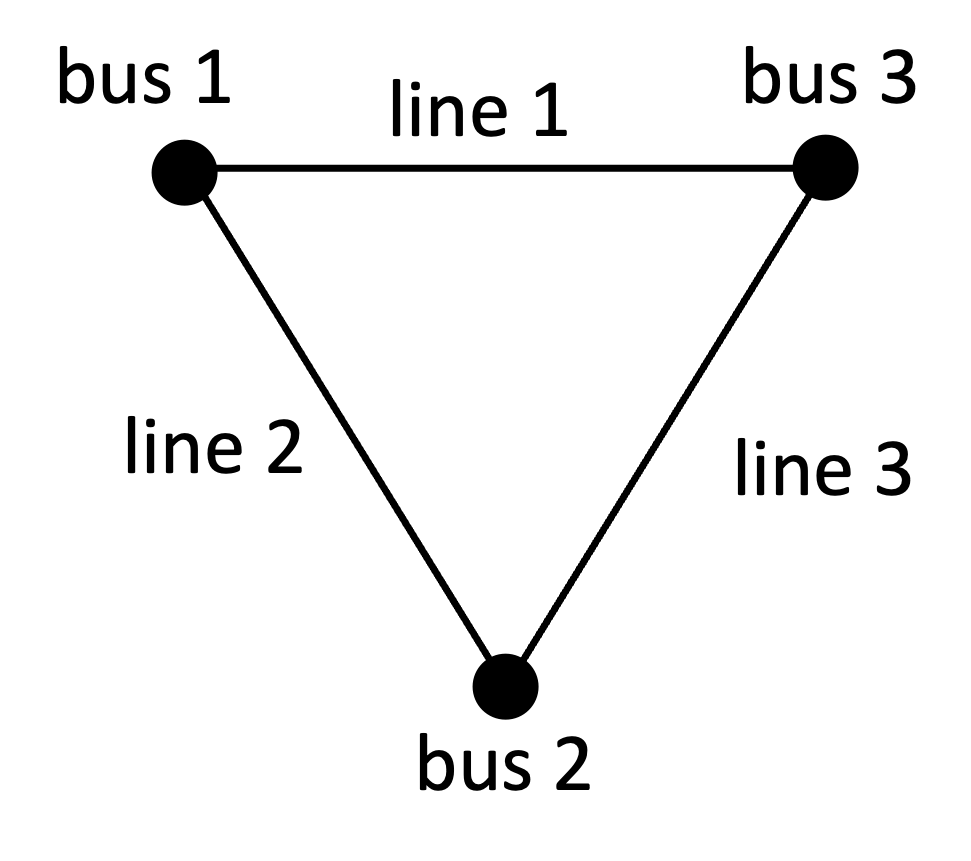}
    \caption{A three-bus meshed network. The indices are $k=\{1,2,3\}$ for the buses and $l=\{4,5,6\}$ for the lines.}
    \label{fig:toy_example_bus3}
\end{figure}

We specify the input and output of this system by definition of modulars given in~\eqref and~\eqref{eq:Modular_Generator_r2}. For each node, the input
$u_1=-p_{12}-p_{13}$, $u_2=-p_{21}-p_{23}$, $u_3=-p_{31}-p_{32}$, and the output $y_1=\delta_1$, $y_2=\delta_2$, $y_3=\delta_3$.

For line 1, define power flow from bus 1 to bus 2 as the positive direction. For line 2, define power flow from bus 1 to bus 3 as the positive direction.  For line 3, define power flow from bus 2 to bus 3 as the positive direction. We associate an index $l=4$, $l=5$, $l=6$ for Line1, Line2, Line3, respectively. Then the input 
\begin{equation}
\begin{split}
    &
        \bm{u}_4=\begin{bmatrix}   \delta_1-\delta_2&  \delta_2-\delta_1 \end{bmatrix}^\top =\underbrace{\left[\begin{matrix}
         1 \\-1
        \end{matrix}\right ]}_{\Phi_{41}} \delta_1
        +
        \underbrace{\left[\begin{matrix}
         -1 \\1
        \end{matrix}\right ]}_{\Phi_{42}} \delta_2\\
    &
        \bm{u}_5=\begin{bmatrix}   \delta_1-\delta_3&  \delta_1-\delta_3 \end{bmatrix}^\top =\underbrace{\left[\begin{matrix}
         1 \\-1
        \end{matrix}\right ]}_{\Phi_{51}} \delta_1
        +
        \underbrace{\left[\begin{matrix}
         -1 \\1
        \end{matrix}\right ]}_{\Phi_{53}} \delta_3\\    
    &
        \bm{u}_6=\begin{bmatrix}   \delta_2-\delta_3&  \delta_3-\delta_2 \end{bmatrix}^\top =\underbrace{\left[\begin{matrix}
         1 \\-1
        \end{matrix}\right ]}_{\Phi_{62}} \delta_2
        +
        \underbrace{\left[\begin{matrix}
         -1 \\1
        \end{matrix}\right ]}_{\Phi_{63}} \delta_3
\end{split}    
\end{equation}

and the output
\begin{equation}
    \begin{split}
&
    \bm{y}_4=\begin{bmatrix}   \frac{\left (g_4-g_4\cos(\delta_1-\delta_2) \right )}{2\alpha_4}+\frac{b_4\sin(\delta_1-\delta_2)}{2}
    \\
    \frac{\left (g_4-g_4\cos(\delta_2-\delta_1) \right )}{2\alpha_4}+\frac{b_4\sin(\delta_2-\delta_1)}{2} \end{bmatrix}\\
&
    \bm{y}_5=\begin{bmatrix}   \frac{\left (g_5-g_5\cos(\delta_1-\delta_3) \right )}{2\alpha_5}+\frac{b_5\sin(\delta_1-\delta_3)}{2}
    \\
    \frac{\left (g_5-g_5\cos(\delta_3-\delta_1) \right )}{2\alpha_5}+\frac{b_5\sin(\delta_3-\delta_1)}{2} \end{bmatrix}\\
&
    \bm{y}_6=\begin{bmatrix}   \frac{\left (g_6-g_6\cos(\delta_2-\delta_3) \right )}{2\alpha_3}+\frac{b_6\sin(\delta_2-\delta_3)}{2}
    \\
    \frac{\left (g_6-g_6\cos(\delta_3-\delta_2) \right )}{2\alpha_6}+\frac{b_6\sin(\delta_3-\delta_2)}{2} \end{bmatrix}
    \end{split}
\end{equation}
From the power flow in the line $l$ connecting bus $i$ and $j$ given by
\begin{equation}\label{eq: power_flow_toy_2}
p_{ij}=g_l-g_l\cos(\delta_i-\delta_j)+b_l\sin(\delta_i-\delta_j),    
\end{equation}
we have
\begin{equation*}
    \begin{split}
&     -p_{12}=  \underbrace{\left[\begin{matrix}
     -1 & 1
    \end{matrix}\right ]}_{\Phi_{14}} \bm{y}_4
    +\underbrace{\alpha_4\left[\begin{matrix}
     -1 & -1
    \end{matrix}\right ]}_{\Psi_{14}} \bm{y}_4
    \text{  and  }\\
   & -p_{21}= \underbrace{\left[\begin{matrix}
     1 & -1
    \end{matrix}\right ]}_{\Phi_{24}} \bm{y}_4
    +\underbrace{\alpha_4\left[\begin{matrix}
     -1 & -1
    \end{matrix}\right ]}_{\Psi_{24}} \bm{y}_4, \\
&     -p_{13}=  \underbrace{\left[\begin{matrix}
     -1 & 1
    \end{matrix}\right ]}_{\Phi_{15}} \bm{y}_5
    +\underbrace{\alpha_5\left[\begin{matrix}
     -1 & -1
    \end{matrix}\right ]}_{\Psi_{15}} \bm{y}_5
    \text{  and  }\\
   & -p_{31}= \underbrace{\left[\begin{matrix}
     1 & -1
    \end{matrix}\right ]}_{\Phi_{35}} \bm{y}_5
    +\underbrace{\alpha_5\left[\begin{matrix}
     -1 & -1
    \end{matrix}\right ]}_{\Psi_{35}} \bm{y}_5, \\
&     -p_{23}=  \underbrace{\left[\begin{matrix}
     -1 & 1
    \end{matrix}\right ]}_{\Phi_{26}} \bm{y}_6
    +\underbrace{\alpha_6\left[\begin{matrix}
     -1 & -1
    \end{matrix}\right ]}_{\Psi_{26}} \bm{y}_6
    \text{  and  }\\
    &-p_{32}= \underbrace{\left[\begin{matrix}
     1 & -1
    \end{matrix}\right ]}_{\Phi_{36}} \bm{y}_6
    +\underbrace{\alpha_6\left[\begin{matrix}
     -1 & -1
    \end{matrix}\right ]}_{\Psi_{36}} \bm{y}_6
    \end{split}
\end{equation*}
To investigate the stability of the whole interconnected system, we stack the input/output vectors in sequence as $\bm{u}:=(u_1, u_2, u_3,\bm{u}_4^\top,\bm{u}_5^\top,\bm{u}_6^\top)$ and $\bm{y}:=(y_1, y_2, y_3, \bm{y}_4^\top, \bm{y}_5^\top, \bm{y}_6^\top)$. Their interconnection is shown in~\eqref{eq:interconnection_bus3}.

\begin{figure*}
\hrule
\small
\begin{equation}\label{eq:interconnection_bus3}
\includegraphics[width=.9\linewidth]{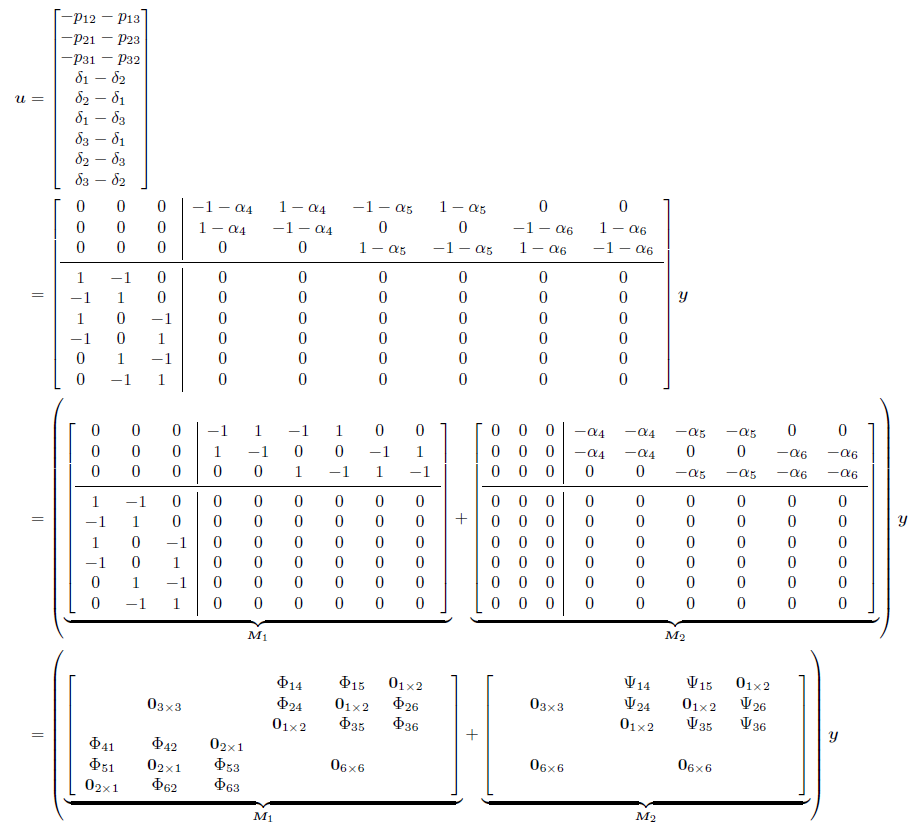}
\end{equation}
\normalsize
\hrule
\end{figure*}

In this case, we again have $\bm{M}_1$ being skew-symmetric and $\bm{M}_2$ being sparse. Since $\bm{M}_1$ and $\bm{M}_2$ come from the definition of input and output of the modules, the interconnection matrix always decouples into two parts. Therefore, $\bm{u}=\bm{M}\bm{y}$ holds for all network topologies. Such modular approach is also one of the main contribution of this paper. That is, the analysis is not impacted by the topology changes and the plug-in and the plug-out of devices.

\subsection{Proof of EIP Condition for Memoryless Systems}\label{app:app:memoryless}
The conclusion can also be found in~\cite{arcak2016networks} but a proof is not not explicitly given. For completeness, we supplement the proof as follows. 
First, we show sufficiency. Obviously, ~\eqref{eq:memoryless_EIP} holds if $u=u^*$. When  $u\neq u^*$, $y_l^\prime(u)\in
[0,\frac{1}{\epsilon_l}]$ gives
$0\leq \frac{(y_l(u)-y_l(u^*)}{u- u^*}\leq\frac{1}{\epsilon_l} $. From $0\leq \frac{(y_l(u)-y_l(u^*)}{u- u^*}$,  $(y_l(u)-y_l(u^*)$ is the same sign as $(u- u^*)$ and therefore $(y_l(u)-y_l(u^*))(u- u^*)\geq 0$. Then multiplying both side of $\frac{(y_l(u)-y_l(u^*)}{u- u^*}\leq\frac{1}{\epsilon_l}$ with $(y_l(u)-y_l(u^*))(u- u^*)$ gives~\eqref{eq:memoryless_EIP}.

Next, we show necessity. Without loss of generality, suppose $u^*<u$. By mean value theorem, there exists a $\Tilde{u}\in[u^*,u]$ such that $(y_l(u)-y_l(u^*)=y_l^\prime(\Tilde{u})(u- u^*)$. Then, $\left(u-u^*\right)\left(y_l(u)-y_l(u^*)\right)-\epsilon_l
\left(y_l(u)-y_l(u^*)\right)^2\geq 0$ gives $y_l^\prime(\Tilde{u})\left(u-u^*\right)^2-\epsilon_l
\left(y_l^\prime(\Tilde{u})\right)^2\left(u-u^*\right)^2\geq 0$. Since $u^*<u$ indicates $(u-u^*)^2>0$, we have $y_l^\prime(\Tilde{u})-\epsilon_l
\left(y_l^\prime(\Tilde{u})\right)^2\geq 0$ and hence $y_l^\prime(\Tilde{u})\in \left [0, \frac{1}{\epsilon_l}  \right ]$. This holds for any $u^*\in\mathcal{U}$, $u\in\mathcal{U}$. Hence,   $y_l^\prime(\Tilde{u})\in \left [0, \frac{1}{\epsilon_l}  \right ]$ for   $\Tilde{u}\in\mathcal{U}$. Namely, $y_l^\prime(u)\in \left [0, \frac{1}{\epsilon_l}  \right ]$ for   $u\in\mathcal{U}$.